\documentclass[reprint,amsfonts, amssymb, amsmath,  showkeys,aps, superscriptaddress, twocolumn,longbibliography,nofootinbib]{revtex4-2}
\usepackage{graphicx} 

\usepackage[dvipsnames]{xcolor}

\usepackage{float}
\usepackage[shortlabels]{enumitem}

\usepackage{braket}
\usepackage{amsthm}
\usepackage{mathtools}
\usepackage{url}
\usepackage{physics}
\usepackage{graphicx}
\usepackage[left=16mm,right=16mm,top=35mm,columnsep=15pt]{geometry} 
\usepackage[T1]{fontenc}
\usepackage{bm}
\usepackage{tabularx}
\newcommand*{\eh}{\mathrm{End\ }\mathcal{H}}
\newcommand*{\Ad}{\mathrm{Ad}}

\def\ad{^{\dagger}}

\def\a{\alpha}

\newcommand{\fsnull}[1]{}
\newcommand{\old}[1]{}

\usepackage[makeroom]{cancel}
\usepackage[toc,page]{appendix}
\definecolor{C1}{RGB}{52, 89, 149}
\definecolor{C2}{RGB}{251, 77, 61}
\definecolor{C3}{RGB}{3, 206, 164}
\definecolor{C4}{RGB}{202, 21, 81}
\usepackage{hyperref}
\hypersetup{colorlinks=true, linkcolor=C4, citecolor=C4, urlcolor=C4}

\usepackage{tikz}
\tikzset{every picture/.style=remember picture}

  \usetikzlibrary{arrows.meta,decorations.pathreplacing}

\usepackage{pgfplots}
\pgfplotsset{compat=1.16}
\pgfplotsset{seriesstyle/.style={line width=1pt, mark=*, mark size=1.2pt}}

\usepackage[utf8]{inputenc}
\usepackage{graphicx}
\usepackage{xcolor}
\usepackage{amsmath}
\usepackage{amsthm}
\usepackage{bm}
\usepackage{bbm}
\usepackage{comment}
\usepackage{appendix}
\usepackage{mathdots}
\usepackage{lipsum}
\usepackage{verbatim}
\usepackage{nccmath}
\usepackage{amsfonts}
\usepackage{thm-restate}
\usepackage{thmtools}
\usepackage{capt-of}
\usepackage{booktabs}

\usepackage{amssymb}
\usepackage{dsfont}

\newcommand{\img}[2][2.2ex]{%
  \mathrel{\vcenter{\hbox{\includegraphics[height=#1]{#2}}}}%
}

\newcommand{\Var}{{\rm Var}}

\renewcommand{\geq}{\geqslant}
\renewcommand{\leq}{\leqslant}

\def\endh{\mathrm{End\ }\mch}

\newcommand{\ot}{\otimes}
\newcommand{\ts}{^{\otimes 2}}

\newcommand{\bs}{\textsf{BS}}

\newcommand{\lm}{\lambda }

\newcommand{\sg}{\sigma }

\DeclareMathOperator*{\expect}{\mathbb{E}}

\newcommand{\mcl}{\mathcal{L}}

\newcommand{\mcw}{\mathcal{W}}
\newcommand{\mco}{\mathcal{O}}

\newcommand{\mcc}{\mathcal{C}}

\newcommand{\mcv}{\mathcal{V}}
\newcommand{\mch}{\mathcal{H}}
\newcommand{\mcm}{\mathcal{M}}
\newcommand{\mcp}{\mathcal{P}}
\newcommand{\mca}{\mathcal{A}}
\newcommand{\mcd}{\mathcal{D}}
\newcommand{\mce}{\mathcal{E}}

\newcommand{\mct}{\mathcal{T}}
\newcommand{\mbsp}{\mathbb{SP}}
\newcommand{\mbso}{\mathbb{SO}}
\newcommand{\mbsu}{\mathbb{SU}}
\newcommand{\mbo}{\mathbb{O}}
\newcommand{\mbu}{\mathbb{U}}

\newcommand{\mbc}{\mathbb{C}}
\newcommand{\mbr}{\mathbb{R}}

\newcommand{\mbs}{\mathbb{S}}
\newcommand{\mbz}{\mathbb{Z}}
\newcommand{\mbf}{\mathbb{F}}
\newcommand{\mbe}{\mathbb{E}}
\newcommand{\mst}{\mathsf{T}}

\newcommand{\mfsp}{\mathfrak{sp}}

\def\be{\begin{equation}}
\def\ee{\end{equation}}
\def\bs{\begin{split}}
\def\e{\end{split}}
\def\ba{\begin{eqnarray}}
\def\bea{\begin{eqnarray}}

\def\tea{\end{eqnarray}}
\def\ea{\end{eqnarray}}
\def\eea{\end{eqnarray}}

\def\d{\delta}

\def\a{\alpha}
\def\b{\beta}

\def\d{\delta}
\def\a{\alpha}

\def\b{\beta}

\def\g{\mathfrak{g}}

\def\tn{^\otimes n}
\def\tk{^\otimes k}

\def\a{\alpha}
\def\b{\beta}

\def\tn{^{\otimes n}}

\def\tk{^{\otimes k}}

\newcommand{\id}{\mathds{1}}

\renewcommand{\a}{\alpha}
\renewcommand{\b}{\beta}
\newcommand{\ga}{\gamma}

\renewcommand{\d}{\delta}

\newcommand{\sbraket}[2]{ \langle#1 | #2 \rangle}

\def\sg{\sigma}

\def\g{\gamma}

\def\be{\begin{equation}}
\def\te{\end{equation}}
\def\ee{\end{equation}}
\def\ba{\begin{eqnarray}}
\def\bea{\begin{eqnarray}}

\def\tea{\end{eqnarray}}
\def\ea{\end{eqnarray}}
\def\eea{\end{eqnarray}}

\begin{document}

\title{Classical shadows over symmetric spaces}

\author{Rebecca Chang}
\affiliation{Massachusetts Institute of Technology, Cambridge, Massachusetts 02139, USA}
\affiliation{Theoretical Division, Los Alamos National Laboratory, Los Alamos, New Mexico 87545, USA}

\author{Maureen Krumt{\"u}nger}
\affiliation{School of Physics, University of Melbourne, Parkville, VIC 3010, Australia}

\author{Mart\'{i}n Larocca}
\affiliation{Theoretical Division, Los Alamos National Laboratory, Los Alamos, New Mexico 87545, USA}
\affiliation{Quantum Science Center, Oak Ridge, TN 37931, USA}

\author{Maxwell West}
\affiliation{Theoretical Division, Los Alamos National Laboratory, Los Alamos, New Mexico 87545, USA}
\affiliation{School of Physics, University of Melbourne, Parkville, VIC 3010, Australia}

\begin{abstract}
Efficiently learning expectation values of unknown quantum states via \textit{classical shadows} has become an important primitive in both theoretical and experimental aspects of quantum computation. Typically, classical shadow protocols involve randomised measurements induced by sampling uniformly randomly from a compact group, a situation which is now quite well understood. In this work we go beyond this standard assumption, studying the classical shadow protocols occasioned by sampling uniformly randomly from the so-called \textit{compact symmetric spaces}. We uncover a unifying theory of such protocols, extending the extent to which the general theory of classical shadows is understood at a mathematical level. Interestingly, for the estimation of observables sampled from certain distributions we further find that some of these protocols allow for slight improvements in sample-complexity over existing shadow schemes. {As part of our technical analysis we obtain $\varepsilon$-approximate relative-error  $k$-designs over the $n$-qubit symplectic group in depth $\mco(nk+\log(1/\varepsilon))$, which is optimal with respect to $n$, and may be of  some independent interest. }
\end{abstract}

\maketitle

\section{Introduction}
Determining properties of unknown quantum states is both a fundamental task in quantum information theory and of rapidly increasing practical relevance. In recent years, as the surging size of experimentally accessible quantum systems has forced the development of techniques beyond full state tomography, the theory of  \textit{classical shadows}~\cite{huang2020predicting} has emerged as a leading framework within which to perform such learning. In a classical shadows protocol one measures samples of an unknown target  state $\rho$ in randomised bases dictated by a protocol-dependent ensemble $\mce$ of unitaries, and uses the resulting data to reconstruct ``shadows'' of $\rho$, which are in turn used as proxies against which observables can be estimated. For a given target set of observables, the (sample-)efficiency with which this can be done  depends strongly on the selected ensemble of unitaries, with an explosion of research into the various benefits of different ensembles having quickly appeared~\cite{huang2020predicting,west2026classical,zhao2021fermionic,wan2022matchgate,low2022classical,west2025real,hearth2024efficient,chen2021robust,bertoni2024shallow,king2024triply,van2022hardware,hu2021classical,bu2024classical,bringewatt2025classical,sauvage2024classical}. 

\begin{figure}
    \centering
    \includegraphics[width=0.98\linewidth]{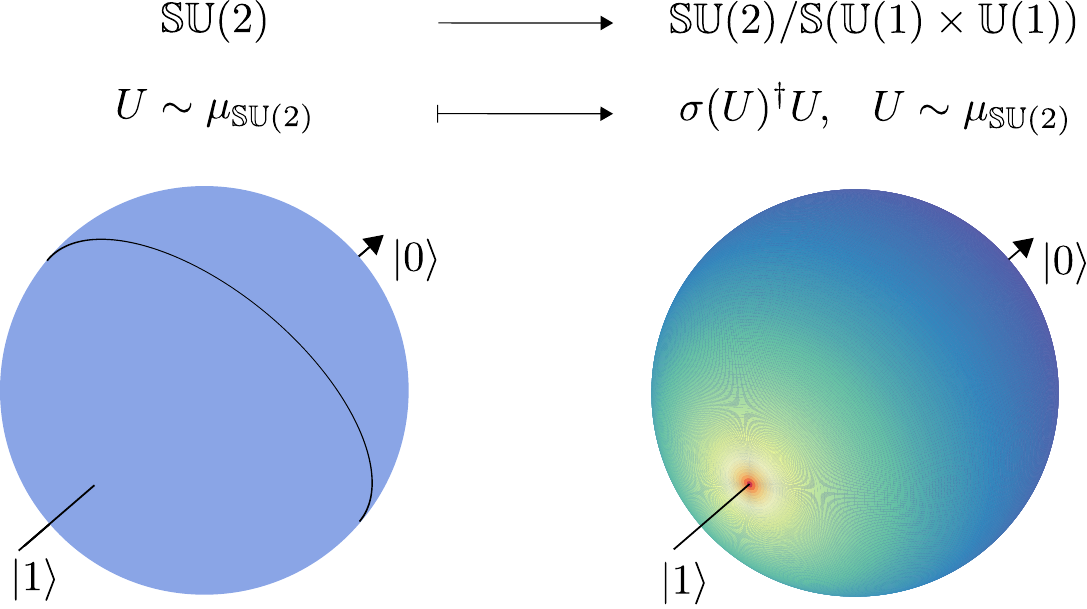}
    \caption{The probability distributions on the Bloch sphere resulting from acting on the state $\ket 0$ with unitaries drawn from the uniform ensembles on $\mbsu(2)$ {(left)} and the AIII symmetric space $\mbsu(2)/\mbs(\mbu(1)\times\mbu(1))$ {(right)}. {In general, the invariant measure $\mu_{G/K}$ on $G/K$ is defined by sampling $g$ from the uniform measure $\mu_G$  on $G$, and then ``returning'' the product $\sg(g)^{-1}g$, where $\sg$ is a certain $G$-homomorphism that defines the symmetric space (see Table~\ref{tab:symm_classif}). This shows that $\mu_{G/K}$ may be thought of as  a certain non-uniform measure on $G$.}  }
    \label{fig:1}
\end{figure}

Although the properties of classical shadow protocols are now quite well understood when the ensemble $\mce$ of unitaries from which one is sampling forms a compact group~\cite{west2026classical}, there has been little exploration of the behaviour of classical shadow protocols beyond that assumption~\cite{bertoni2024shallow}. In this work we address this, studying the classical shadow protocols induced by sampling uniformly randomly from the seven infinite families of \textit{compact symmetric spaces of type I}~\cite{knapp1996lie}. Such symmetric spaces can be thought of as a quotient $G/K$, with $K$ a (generically non-normal) subgroup of a compact Lie  group $G$ satisfying various conditions. 
{For example, the symmetric space $\mbu(d)/\mbo(d)$ corresponds to the manifold of symmetric (in the sense of satisfying $X^\mst=X$) unitary matrices of dimension $d\times d$ . }
Symmetric spaces have long been recognised to play a fundamental role in the analysis of symmetry in quantum mechanics~\cite{dyson1962statistical,dyson1962threefold}; more recently, they have found technical applications in quantum compilation~\cite{wierichs2025recursive,khaneja2001cartan,mansky2023near,dagli2008general,kottmann2026parameter}.  As we shall see, the classical shadow protocol induced by sampling uniformly from $G/K$ may be interpreted as the protocol induced by sampling from a certain non-uniform distribution on $G$ (see Fig.~\ref{fig:1}). {The fact that the manifold $G/K$ does not generally possess a group structure   complicates the technical analysis compared to the case of group-based shadows somewhat; despite this, we are nonetheless able to  uncover a unifying theory of such protocols.} 
Indeed, we find that the classical shadow \textit{measurement channel} associated to $G/K$ is a convex combination of the channel associated to $G$, a dephasing channel, and (when $G$ is the symplectic group) a further subleading term. Essentially, therefore, the classical shadow protocols on $G/K$ resemble those on $G$, but modified to prefer a certain basis. Interestingly, some of these protocols allow for a non-vanishing improvement in sample-complexity over the group-based ensembles for the estimation of observables sampled from distributions centred on the preferred basis, and may therefore find experimental applications.

\begin{table*}
    \centering
    \begin{tabular}{lccccc}
        Type & $G$ &   $K$ &    $\sigma$ &  $\a_{G/K}$ &  $|\a_{G/K}-\b_{G/K}|$ \\ \midrule
        AI & $\mbu(d)$  & $\mbo(d)$  & $*$ & $2/(d(d+3))$&0 \\
        AII & $\mbsu(d)$  & $\mbsp(d/2)$  & $\Ad_J \circ *$ & $2/(d(d-1))$&0\\
        AIII & $\mbsu(p+q)$ &  $\mbs(\mbu(p)\times \mbu(q))$  &$\Ad_{I_{p,q}} $& $(s^4  + 2 s^2 (d-2)  + d^2)/(d^2(d-1)(d+3))$&0\\
        \midrule
        BDI & $\mbso(p+q)$ &  $\mbso(p)\times \mbso(q)\ $  & $\Ad_{I_{p,q}} $ &$(  s^{4}  + (6d-4)s^2  + 3 d(d - 2))/(d(d^2-1)(d+6))$&0\\
        DIII & $\mbso(2d)$  & $\mbu(d)$ & $\Ad_{J} $& $3/(d^2-1)$&0 \\
        \midrule
        CI & $\mbsp(d)$  & $\mbu(d)$  & $\Ad_{J} $ & $  3/((d-1)(d+3))$&$\mco(1/d^3)$\\
        CII & $\mbsp(p+q)$ & $\mbsp(p)\times \mbsp(q)$ &$\Ad_{K_{p,q}} $ & $(4s^4  - 10 s^2  + 3(d/2)^2 + 3d/2)/(d/2(d/2-1)(d-1)(d+3))$&$\mco(1/d)$\\
    \end{tabular}
    \caption{Classification of the  classical compact symmetric spaces of type I.
    Here $G$ is a ``parent'' group with a subgroup $K$  the fixed points of an involution $\sigma:G\to G$, and $G/K$ the corresponding symmetric space. $J$ is the canonical symplectic form, $I_{p,q}=\id_p \oplus (-\id_q)$, and $K_{p,q} = I_{p,q}^{\oplus 2}$.  The classical shadow channel induced by sampling uniformly from $G/K$  and measuring in the basis $\mcw$ with respect to which the involutions take their standard forms (as given in the Table) is a convex combination $\mcm_{G/K,\mcw}(\rho) = (1-\a_{G/K})\mcm_{G,\mcw}(\rho) + \b_{G/K}\mca_{\mcw}(\rho)+(\a_{G/K}-\b_{G/K})\left(J\mca_\mcw(\rho)J\ad - \mca_\mcw(\rho J)J\right)$ of the channel corresponding to sampling uniformly from $G$, a dephasing channel, and (in the symplectic case) a subleading term involving the symplectic form. For AIII, BDI and CII, $s$ is the signature of the matrix whose adjoint action defines the involution. 
    \label{tab:symm_classif}}
\end{table*}

\section{Preliminaries}\label{sec:prelim}
We begin by briefly recalling the key points of classical shadows; more details may be found in Appendix~\ref{sec:cs} and Refs.~\cite{huang2020predicting,west2026classical}. A classical shadow protocol $(\mce,\mcw)$ is defined by two choices: an ensemble $\mce$ of unitaries acting on the Hilbert space $\mch\cong\mbc^d$ of the unknown state $\rho$, and a measurement basis $\mcw=\{\ket{w}\}_w$ of $\mch$. Having made these choices, one takes copies of one's unknown state, acts with a {$U$ randomly sampled from $\mce$, which we denote by $U\sim\mce$}, and measures in the basis $\mcw$. Upon obtaining the outcome (say) $\ket w$, one classically stores the state $U\ad \ketbra{w}{w} U$. The result of all this is to effect the \textit{measurement channel}\footnote{Note that this is just $\expect_{g\sim \mce} \, g\cdot \mca_\mcw$, where $\cdot$ is the natural action of $g\in\mce$ on $\rm{End}(\rm{End}(\mch))${, and $\rm{End}(V)$ denotes the space of linear operators on a vector space $V$.}}
\begin{equation}\label{eq:mc}
\mcm_{\mce,\mcw} = \expect_{U\sim \mce} {\rm Ad}_{U^\dagger}\circ \mca_\mcw \circ {\rm Ad}_{U},
\end{equation}
where $\mca_\mcw(-)=\sum_w \bra{w}(-)\ket{w}\ketbra{w}$ is a dephasing channel with respect to $\mcw$ and $\Ad_U(-) = U(-)U\ad$ the adjoint action of $U$. The (pseudo-)inverse of $\mcm_{\mce,\mcw}$ on the states generated in this way from a collection of $N$ copies of $\rho$ forms a set $\{\mcm_{\mce,\mcw}^{-1}(U_i\ad \ketbra{w_i}{w_i} U_i)\}_{i=1}^N$ of \textit{classical shadows} of $\rho$ which one uses as proxies $\{\hat{\rho}_i\}_{i=1}^N$ against which to evaluate expectation values of observables $O$,
\begin{equation}\label{eq:emp}
    \tr[\rho O] \approx \frac{1}{N}\sum_{i=1}^N \tr[\hat{\rho}_i O].
\end{equation}
Two natural questions arise: (i) are the estimators of Eq.~\eqref{eq:emp} unbiased, and (ii) how many shadows does one need to (with a certain probability) obtain estimates to a given target precision? We omit a detailed discussion of these well-studied questions here (but see Appendix~\ref{sec:cs} or e.g. Refs.~\cite{huang2020predicting,west2026classical}); suffice it to say that their answers depend on the set of observables which one is interested in measuring, and that (conditioned on the selection of observables) different choices of $\mce$ and $\mcw$ can lead to protocols of varying favorability. For example, sampling from the uniform distribution on the matchgate group~\cite{braccia2025optimal} allows for the efficient estimation of low-degree fermionic observables~\cite{wan2022matchgate,zhao2021fermionic,low2022classical}, and uniformly sampling local Cliffords leads to the efficient estimation of operators with constant support~\cite{huang2020predicting}. As foreshadowed, the point of this work is to investigate the effects of taking $\mce$ to be the uniform measure on a classical compact symmetric space, some relevant properties of which we now review (see Appendix~\ref{sec:ss} for more details).  

One can speak of symmetric spaces at various levels of generality and from various points of view~\cite{cartan1926sur,gorodski2021introduction,magnea2002introduction,zirnbauer2010symmetry}; for our purposes it will be sufficient to make the
\begin{restatable}{defn}{ss}
    A classical compact symmetric space (of type I) is a quotient $G/K$, where $G$ is the unitary, orthogonal or (unitary) symplectic group, and $K$ is  the fixed point set of an involution $\sg:G\to G$.
\end{restatable}
For example, take $G=\mbu(d)$ with the involution $\sg(U)=U^*$. Evidently the fixed point set of $\sg$ is the orthogonal group, and we meet the so-called \textit{AI symmetric space}, ${\rm AI}\equiv\mbu(d)/\mbo(d)$. Remarkably, all of the possible spaces of this form were classified by Cartan~\cite{cartan1926sur} in 1926; there are seven infinite families, tabulated in Table~\ref{tab:symm_classif}. Importantly for our purposes, the uniform measure on a compact symmetric space of the above form has a simple algorithmic description~\cite{duenez2004random,matsumoto2013weingarten}; a random $t\sim G/K$ is given by
$\sg(g)^{-1}g$, where $g$ is drawn uniformly randomly from the ambient group $G$. For example,  a random element of AI can be produced by generating a random $U\in\mbu(d)$ and then ``returning'' $\sg_{\rm AI}(U)\ad U = (U^*)\ad U = U^\mst U$.  What we will in effect  be doing, then, is studying the classical shadow protocols induced by certain non-invariant measures on the unitary, orthogonal, and unitary\footnote{We will for conventional brevity henceforth drop this qualifier, referring to the unitary symplectic group simply as the symplectic group} symplectic groups (see Fig.~\ref{fig:1}).

\section{Results}
As alluded to in the introduction, the application of tools from representation theory has allowed the classical shadow protocols resulting from sampling uniformly from a compact group $G$ to become quite well understood~\cite{west2026classical}. For example, an elementary but powerful observation is that in this setting the measurement channel Eq.~\eqref{eq:mc} is manifestly $G$-equivariant: for any $g\in G,\ \mcm_{G,\mcw}\circ {\rm Ad}_g = {\rm Ad}_g \circ\mcm_{G,\mcw} $. It immediately follows from Schur's lemma~\cite{fulton1991representation} that one has the decomposition
\begin{equation}
    \mcm_{G,\mcw} \cong \bigoplus_{\lm} \id_{d_\lm}\ot s_\lm\,,
\end{equation}
where the direct sum is over the irreps of $G$ in $\eh$. Here $d_\lm$ is the dimension of the irrep $\lm$ (which appears in the decomposition with some multiplicity $m_\lm$) and the $s_\lm \in {\rm End\, }(\mbc^{m_\lm})$ are to-be-determined operators. 

In contrast to this excellent group-structure-induced behaviour, we show in Appendix~\ref{sec:cs} that when we take $\mce$ to be the uniform measure on a symmetric space it is \textit{not} the case that the resulting measurement channel commutes with the action of an arbitrary coset representative in $G/K$\footnote{Where we recall that the $K$ in $G/K$ is generally not a normal subgroup of $G${, so that $G/K$ is generally not a group}}; a somewhat weaker result, however, is true.
To see it, let us say a group $H$ is $\mcw$\textit{-normalising} if, for all $h\in H$ and $\ket{w}\in\mcw$, there exists a $\ket{w'}\in\mcw$ such that $h\ketbra{w}h\ad=\ketbra{w'}$. We will use $N_\mcw$ to denote the group of all $\mcw$-normalising elements of $G$.  We find: 
\begin{restatable}{lem}{lemstab}\label{lem:stab}
    The classical shadows measurement channel $\mcm_{G/K,\mcw}$ resulting from sampling uniformly from $G/K$ and measuring in a basis $\mcw$ commutes with the adjoint action of any group $H\subseteq K\cap  N _\mcw$, so that
    \begin{equation}\label{eq:symstruct}
    \mcm_{G/K,\mcw} = \bigoplus_{\lm} \id_{d_\lm}\ot s_\lm;
\end{equation}
where the direct sum is over the irreps of $H$ in $\eh$. 
\end{restatable}
As we shall soon see, we further find that for all our examples the decomposition of Eq.~\eqref{eq:symstruct} is \textit{multiplicity-free}; that is, $s_\lm\in\mbr\ \forall \lm$. This allows one to trivially invert the measurement channels, avoiding the \textit{a priori} need for extensive classical post processing~\cite{west2026classical}. These calculations are facilitated by the simple observation that the measurement channel Eq.~\eqref{eq:mc} can be re-expressed as 
\begin{equation}\label{eq:mc2}
\mcm_{\mce,\mcw}(\rho)  = \tr_1\left[\left(\rho \otimes \id\right) \mct^{(2)}_{\mce} \left( \sum_{w\in\mcw} \ketbra{w}\ts \right)\right],
\end{equation}
where $\mct^{(k)}(A)= \expect_{U\sim \mce}U\tk A (U\ad)\tk$ for $A\in {\rm End}\, \mch\tk$.
Our immediate task, then,  is an exercise in evaluating second-order twirls over symmetric spaces. 
\begin{figure*}[t]
\centering
\begin{tikzpicture}
\begin{axis}[
  width=9cm, height=7.2cm,
  xlabel={$\|\mathcal{P}_{\rm diag}(O)\|_2^2$},
  ylabel={Variance},
  xmin=0.17, xmax=0.93,
  ymin=0, ymax=2,
  xtick style={draw=none},
  ytick style={draw=none},
  tick label style={/pgf/number format/fixed},
  legend pos=north east,
  legend cell align=left,
  line join=round,
  title={AIII}
]

\addplot+[
  seriesstyle, color=Cerulean, mark options={fill=Cerulean},
  error bars/.cd,
    y dir=both, y explicit,
    error bar style={line width=0.6pt}
] coordinates {
  (0.2, 1.414042255497238) +- (0, 0.02503355885394922)
  (0.3, 1.244160746097832) +- (0, 0.02194189882611683)
  (0.4, 1.07437234643407) +- (0, 0.01885647009287487)
  (0.5, 0.9046592901088768) +- (0, 0.01577642340437677)
  (0.6, 0.7350169098636153) +- (0, 0.01270197475101224)
  (0.7, 0.5654498729569235) +- (0, 0.009634567786507406)
  (0.8, 0.3959759457858749) +- (0, 0.006578609388595113)
  (0.9, 0.2266480380333273) +- (0, 0.003552027186310831)
};
\addlegendentry{$c=0.4$}

\addplot+[
  seriesstyle, color=YellowOrange, mark options={fill=YellowOrange},
  error bars/.cd,
    y dir=both, y explicit,
    error bar style={line width=0.6pt}
] coordinates {
  (0.2, 1.463973971713098) +- (0, 0.0244947600361997)
  (0.3, 1.282538675089919) +- (0, 0.0214695801784342)
  (0.4, 1.101158284706203) +- (0, 0.01843931290428231)
  (0.5, 0.9198223238266713) +- (0, 0.01540494494161877)
  (0.6, 0.7385280401975606) +- (0, 0.0123667545415391)
  (0.7, 0.5572781860726351) +- (0, 0.009324534594005906)
  (0.8, 0.3760832381871722) +- (0, 0.00627744839661603)
  (0.9, 0.1949743970563149) +- (0, 0.003223220753882054)
};
\addlegendentry{$c=0.6$}

\addplot+[
  seriesstyle, color=ForestGreen, mark options={fill=ForestGreen},
  error bars/.cd,
    y dir=both, y explicit,
    error bar style={line width=0.6pt}
] coordinates {
  (0.2, 1.848019544681399) +- (0, 0.03151210524981526)
  (0.3, 1.617526762123244) +- (0, 0.02757342438603409)
  (0.4, 1.387052714716938) +- (0, 0.02363571193315262)
  (0.5, 1.156593827582317) +- (0, 0.01969878344918913)
  (0.6, 0.9261491615931051) +- (0, 0.01576259120660047)
  (0.7, 0.6957196558755779) +- (0, 0.01182718594083494)
  (0.8, 0.4653088853098989) +- (0, 0.007892759962427635)
  (0.9, 0.2349274961613141) +- (0, 0.003959906818541255)
};
\addlegendentry{$c=0.8$}

\addplot+[
  solid, thick, mark=*, mark size=1.2pt, color=Orchid, mark options={fill=Orchid},
  error bars/.cd,
    y dir=both, y explicit,
    error bar style={line width=0.6pt}
] coordinates {
  (0.2, 1.465857306321272) +- (0, 0.02561377795685967)
  (0.3, 1.314130385299702) +- (0, 0.02272629408643288)
  (0.4, 1.162399182188726) +- (0, 0.01984145378756267)
  (0.5, 1.010664514059728) +- (0, 0.01696245261747819)
  (0.6, 0.8589265955585635) +- (0, 0.01409446588214486)
  (0.7, 0.7071852120393777) +- (0, 0.01124771901012105)
  (0.8, 0.5554395464307847) +- (0, 0.008446392043632994)
  (0.9, 0.4036871654320038) +- (0, 0.00576312364945118)
};
\addlegendentry{$\mathbb{U}(d)$}
\end{axis}
\end{tikzpicture}
\hspace{0.35cm}%
\begin{tikzpicture}
\begin{axis}[
  width=9cm, height=7.2cm,
  xlabel={$\|\mathcal{P}_{\rm diag}(O)\|_2^2$},
  ylabel={},
  xmin=0.17, xmax=0.93,
  ymin=0, ymax=2,
  xtick style={draw=none},
  ytick style={draw=none},
  tick label style={/pgf/number format/fixed},
  legend pos=north east,
  legend cell align=left,
  line join=round,
  title={BDI}
]

\addplot+[
  seriesstyle, color=Cerulean, mark options={fill=Cerulean},
  error bars/.cd,
    y dir=both, y explicit,
    error bar style={line width=0.6pt}
] coordinates {
  (0.2, 1.013643376514756) +- (0, 0.01888946994482837)
  (0.3, 0.8926170030706957) +- (0, 0.01657669429894873)
  (0.4, 0.771542359465639) +- (0, 0.01426148655693413)
  (0.5, 0.6504286561955797) +- (0, 0.01194461027889822)
  (0.6, 0.5292783128713597) +- (0, 0.009626662536306811)
  (0.7, 0.4080889098821371) +- (0, 0.007308553522233079)
  (0.8, 0.2868512367319184) +- (0, 0.004992797170962673)
  (0.9, 0.1655378638625367) +- (0, 0.002691622183651479)
};
\addlegendentry{$c=0.4$}

\addplot+[
  seriesstyle, color=YellowOrange, mark options={fill=YellowOrange},
  error bars/.cd,
    y dir=both, y explicit,
    error bar style={line width=0.6pt}
] coordinates {
  (0.2, 1.016151365062286) +- (0, 0.01930889476562874)
  (0.3, 0.8907814455604819) +- (0, 0.0169029309932146)
  (0.4, 0.7654049681392651) +- (0, 0.01449784124025312)
  (0.5, 0.6400231841242828) +- (0, 0.01209348007244145)
  (0.6, 0.5146364222406232) +- (0, 0.009689837235941426)
  (0.7, 0.3892443537631982) +- (0, 0.007287029829837113)
  (0.8, 0.2638457273663598) +- (0, 0.004885440812902998)
  (0.9, 0.1384368165070788) +- (0, 0.002486626939817148)
};
\addlegendentry{$c=0.6$}

\addplot+[
  seriesstyle, color=ForestGreen, mark options={fill=ForestGreen},
  error bars/.cd,
    y dir=both, y explicit,
    error bar style={line width=0.6pt}
] coordinates {
  (0.2, 1.211959024096467) +- (0, 0.02390173567684162)
  (0.3, 1.060966546360478) +- (0, 0.02090615394727691)
  (0.4, 0.9099912340554445) +- (0, 0.01791352984982605)
  (0.5, 0.7590298118218116) +- (0, 0.01492329971609723)
  (0.6, 0.6080814192178072) +- (0, 0.01193531666127494)
  (0.7, 0.4571469166852035) +- (0, 0.008949732382494508)
  (0.8, 0.306229579583555) +- (0, 0.00596712255469189)
  (0.9, 0.1553391621829358) +- (0, 0.002989230616531236)
};
\addlegendentry{$c=0.8$}

\addplot+[
  solid, thick, mark=*, mark size=1.2pt, color=Orchid, mark options={fill=Orchid},
  error bars/.cd,
    y dir=both, y explicit,
    error bar style={line width=0.6pt}
] coordinates {
  (0.2, 1.05087796123399) +- (0, 0.02131778901021711)
  (0.3, 0.9357865838304995) +- (0, 0.01879111070032962)
  (0.4, 0.8205866298461906) +- (0, 0.01625989764267347)
  (0.5, 0.7052988169281971) +- (0, 0.01372626300524047)
  (0.6, 0.5899285876328879) +- (0, 0.0111926185642441)
  (0.7, 0.474470499403894) +- (0, 0.008663638728527425)
  (0.8, 0.358903834594082) +- (0, 0.006152408511995994)
  (0.9, 0.2431668944732698) +- (0, 0.003713625804194467)
};
\addlegendentry{$\mathbb{O}(d)$}
\end{axis}
\end{tikzpicture} 
    \caption{Variances for random (symmetric) observables of unit 2-norm on a $d=128$-dimensional system, averaged over {5}00 instances {(with 95\% confidence intervals indicated, but small enough to be hard to see)}. {In a given instance we sample random matrices $X,Y$ with entries uniformly random on $[0,1)$, and set the initial states to be given by $\rho=XX^\mst/\tr[XX^\mst]$. For the random observable we begin by taking $O=Y+Y^\mst$, and then rescale the diagonal (with respect to the measurement basis $\mcw$) and off-diagonal components (while keeping $\|O\|_2=1$). By choosing different rescaling factors we change the relative weight of the diagonal component of $O$ (the horizontal axis in the graphs). } Here $c=(p-q)/d$ is the (normalised) signature of $I_{p,q}$; each choice of $c$ identifies distinct AIII and BDI protocols. When the   observables are strongly concentrated on the diagonal, AIII and BDI shadows narrowly outperform the standard unitary and orthogonal protocols, respectively. The orthogonal protocols having a variance  lower than their unitary counterparts by a constant factor (of order 2) is expected for symmetric observables~\cite{west2025real}.   }
    \label{fig:var}
\end{figure*}
\noindent
It follows from the discussion of Section~\ref{sec:prelim} that one can express these integrals as \textit{fourth}-order integrals over the ambient group, thus reducing the task to  standard Weingarten calculus~\cite{mele2023introduction} on the unitary, orthogonal and symplectic groups~\cite{collins2006integration}. 
Alternately, it was shown by Matsumoto~\cite{matsumoto2013weingarten} that one can integrate over the symmetric spaces directly, without explicitly passing to the parent group; we explore both of these options in the Appendices. 
In any event, we find:
\begin{restatable}{thm}{thmmain}\label{thm:thmmain}
    Sampling uniformly from $G/K$ and measuring in the basis $\mcw$ with respect to which the involutions take their standard forms yields
    \begin{align}
        \mcm_{G/K,\mcw}(\rho) &= (1-\a_{G/K})\mcm_{G,\mcw}(\rho) + \b_{G/K}\mca_{\mcw}(\rho)\nonumber\\ 
        &\hspace{1mm}+(\a_{G/K}-\b_{G/K})\left(J\mca_\mcw(\rho)J\ad - \mca_\mcw(\rho J)J\right),\label{eq:thm1}
    \end{align}
    where $\mca_{\mcw}$ is a dephasing channel onto $\mcw$ and $J$ is the canonical symplectic form.
    For AI, AII, CI, and DIII, $\a_{G/K}\in\mco(d^{-2})$; for AIII, BDI and CII, $0\le \a_{G/K} \le 1$ is tunable. Finally, $\a_{G/K}=\b_{G/K}$ for $G\neq \mbsp$, and $|\a_{G/K}-\b_{G/K}|\in\mco(1/d)$ for $G=\mbsp$. 
\end{restatable}
Let us make several comments.  
Firstly, it is somewhat gratifying that, despite the lack of a group structure, the shadow channel over any of these symmetric spaces may be written in terms of only two \textit{a priori} unknown parameters (and indeed for $G\neq \mbsp$ we have $\a_{G/K}=\b_{G/K}$, so that there is only one unknown parameter). This observation obviates the need for performing the full calculation of Eq.~\eqref{eq:mc2} explicitly, as one can fix the unknown coefficients by looking at some specific matrix entries. This is a non-trivial simplification; for example, naively evaluating Eq.~\eqref{eq:mc2} over $G/K$ for $G=\mbo,\mbsp$ by performing the fourth-order twirl over $G$ leads to an expression with $((2\cdot 4-1)!!)^2 = 11025$ terms, which must then be symbolically summed over $w$. While this is not impossible, and we indeed provide the results of such calculations in the appendices, Theorem~\ref{thm:thmmain} provides us with an easier route.

Secondly,  from their respective coefficients $\a_{G/K}$ being all of order $d^{-2}$, it follows immediately from Theorem~\ref{thm:thmmain} that the classical shadow protocols induced by AI, AII, CI, and DIII are unlikely to be of practical interest, except insofar as they inherit the benefits of the known protocols over the unitary~\cite{huang2020predicting}, orthogonal~\cite{west2025real}, and symplectic~\cite{west2024random} groups  which they approximately reproduce.

Separately, and though not entirely immediate, it follows from the twin facts that the symplectic group induces unitary state $k$-designs for all $k$~\cite{west2024random} and (as we find in Appendix~\ref{sec:cs}) that $\a_{\rm CII}-\b_{\rm CII}\in\mco(1/d)$  that the AIII and CII protocols coincide to leading order. 
The prospect of finding something particularly new, therefore, rests with the cases of AIII and BDI.

To that end, our first observation is that Eq.~\eqref{eq:thm1} implies that (excepting the degenerate case $\a_{G/K}=1$) all of the measurement channels inherit the images of the channels on the parent group, and are therefore invertible exactly when those channels are. This implies that the set of observables for which they return unbiased estimates does not change~\cite{van2022hardware}. As alluded to in Section~\ref{sec:prelim}, then, the only remaining question is the sample efficiency with which this can be achieved. This reduces to a question of the variance of the resulting estimators (see Appendix~\ref{sec:cs} and Refs.~\cite{huang2020predicting,west2026classical}), with lower variances leading to improved sample complexities. In the cases of sampling uniformly from the unitary, orthogonal and symplectic groups, it has been shown in Refs.~\cite{huang2020predicting,west2025real,west2024random} respectively that the variance is (essentially) controlled by the 2-norm of the observable to be estimated. While the precise expressions for the variances are in the symmetric space case quite unwieldy\footnote{For example in the BDI case the exact expression involves $(11!!)^2=108,056,025$ terms} (see Appendix~\ref{sec:cs}) one can intuit from Eq.~\eqref{eq:thm1} a dependence on the 2-norm of the projection of the observable onto the space of operators which are diagonal in the basis of the dephasing channel; indeed, we find this behaviour numerically in Fig.~\ref{fig:var} (see Appendix~\ref{sec:cs} for an analytic bound). For observables with non-negligible concentration on the diagonal subspace (again, with respect to the basis of the dephasing channel) we further find that the AIII and BDI protocols offer modest improvements over the unitary and orthogonal protocols respectively.

\section{Discussion}
The development of classical shadow protocols employing novel unitary ensembles has become a popular line of research, with many instances enjoying various advantages now known~\cite{huang2020predicting,zhao2021fermionic,wan2022matchgate,low2022classical,west2025real,hearth2024efficient,chen2021robust,bertoni2024shallow,king2024triply,west2026classical,van2022hardware,hu2021classical,bu2024classical,bringewatt2025classical,sauvage2024classical}. {For example, it is common to attempt to find an ensemble which leads to the efficient estimation of observable satisfying  certain symmetries~\cite{west2026classical,sauvage2024classical} (with the use of the word symmetries here not to be conflated with the rather different usage in ``symmetric spaces'')}. Given the results presented here, we are now in the satisfying situation of the protocols induced by both the (defining representations of the) classical compact groups~\cite{huang2020predicting,west2025real,west2024random} and symmetric spaces being understood. \\

One important aspect of the shadows pipeline that we have left hitherto undiscussed, however, is the circuit complexity of actually sampling from our ensembles. 
{Recalling that our ensemble elements look like $\sigma(g)^{-1}g$ for $\sigma$ the symmetric space involution and $g$ a   random element of the parent group $G$, there are two questions: How hard is it to implement the involution $\sigma$, and  how  hard is it to generate a random element of $G$?  Let us begin by considering the difficulty of implementing the involution  $\sigma$. Consulting Table~\ref{tab:symm_classif}, we see that this will variously involve  transposing the random element $g$, and applying the symplectic form $J$, or the operators  $I_{p,q}=\id_p \oplus (-\id_q)$, and $K_{p,q} = I_{p,q}^{\oplus 2}$. Transposing $g$ is easy if we can compile\footnote{{The question of actually performing such a compilation for  a random $g$ is discussed presently}} it into two-qubit unitaries:   one then (readily) transposes these small building blocks, and applies the results in    reverse  order  (as $(AB)^\mst = B^\mst A^\mst$). Meanwhile, the canonical symplectic form is simply given by applying $-iY$ to the first qubit, $I_{p,q}$ can be applied in logarithmic depth by means of the compilation of Ref.~\cite{vandaele2026asymptotically}, and $K_{p,q}=\id_2\ot I_{p,q}$ is simply given by an identity on the first qubit and an $I_{p,q}$ on the remaining $n-1$ qubits. The remaining question, then, is the cost of implementing a random element of $G$ (compiled down into two-qubit gates)}.
As is well-known~\cite{huang2020predicting} (and as we see in the appendices) it is actually somewhat overkill to be sampling exactly from the uniform measures on the symmetric spaces; indeed, any 3-design over those measures will do. As we have further seen that the moments of these ensembles can be written in terms of doubled-order moments of the parent group, it follows that it suffices for our purposes to sample variously from 6-designs over the unitary, orthogonal and symplectic groups. Counterintuitively, {when using circuits composed of local gates acting on a one-dimensional lattice} this can be done in logarithmic depth for the unitary group~\cite{schuster2024random,laracuente2024approximate}, but not  in sublinear-depth for the orthogonal or symplectic groups~\cite{west2026ambient,west2025no,grevink2025will}.
{In the case of the orthogonal group, a matching upper bound (that is, a linear-depth design construction) is known~\cite{baer2026random}; interestingly, however, a similar result for the symplectic group does not appear to have been proven. Remedying this, we prove the existence of linear-depth symplectic designs in Appendix~\ref{sec:symp}; more precisely, we show that we can obtain (relative-error~\cite{mele2023introduction}) $\varepsilon$-approximate symplectic $k$-designs in depth $\mco(nk+\log(1/\varepsilon))$. Putting the pieces of the last few paragraphs together, it follows that we can sample from our target ensembles in depth $\mco(\log n)$ when $G$ is the unitary group, and depth $\mco( n)$ otherwise. So, either way our constructions are  efficient with respect to circuit depth.}

Throughout this work, we have let the symmetric spaces act on the Hilbert space of the unknown state in the manner conferred by the defining representations of the unitary, orthogonal and symplectic groups; a natural extension would be to consider the protocols induced by arbitrary representations. A particularly natural example of this would be to consider DIII with $\mbso$ acting as the group of matchgate circuits; the quotient $\mbso(2n)/\mbu(n)$ is the manifold of inequivalent pure Gaussian states~\cite{hackl2021bosonic}. In this more general setting (a direct generalisation of) Lemma~\ref{lem:stab} would still hold, in principle allowing for at least partial understanding of the measurement channel with only minimal explicit calculation.
The primary shortcoming of this approach (which indeed exists already in the case of the defining representations, whence the necessitation of the detailed calculations of the appendices)  at its current level of development is its inability to directly identify the operators $s_\lm$ of Eq.~\eqref{eq:symstruct}. This is a situation which has been (partially) resolved in the case of shadows over a compact group; modulo mild requirements on the measurement basis, a simple analytical expression has been found for the corresponding coefficients, obviating the need for explicit Weingarten calculus~\cite{west2026classical}. The discovery of a similar expression in the symmetric space case would thus be very interesting. 

\section{Acknowledgments}
We gratefully acknowledge useful discussions with Nathan Killoran. 
This work was supported by Laboratory Directed Research and Development (LDRD) program of Los Alamos National Laboratory (LANL) under project number 20230049DR.
RC was supported by the U.S. Department of Energy through a quantum computing program sponsored by the Los Alamos National Laboratory Information Science \& Technology Institute.
MK acknowledges support from the Melbourne Research Scholarship.
ML acknowledges support from LANL's ASC Beyond Moore’s Law project and the Quantum Science Center (QSC), a national quantum information science research center of the U.S. Department of Energy (DOE).
MW acknowledges the support of the Australian government research training program scholarship, the Dr. Albert Shimmins Fund,  and the IBM Quantum Hub  at the University of Melbourne. 

\bibliography{refs,quantum}

\onecolumngrid
\appendix

\section{Integration over symmetric spaces}\label{sec:ss}
In this appendix we discuss the basic details of integrating over (classical, compact, type I) symmetric spaces; more details may be found in Refs.~\cite{duenez2004random,matsumoto2013weingarten}. We recall that such a symmetric space may be explicitly realised as a homogeneous space $G/K$, where $G$ is one of the matrix Lie groups $\mbu(d)$, $\mbso(d)$, or $\mbsp(d)$, and $K \subset G$ is  the fixed–point set of an involutive automorphism $\sigma:G\to G$, i.e.\ $K = \{g\in G:\sigma(g)=g\}$.
Happily, the uniform measure\footnote{A rigorous discussion of the existence and uniqueness of which we omit, but see e.g. Ref.~\cite{duenez2004random}} on $G/K$ admits a simple sampling recipe: sample $g\in G$ uniformly randomly and then ``return'' $\sg(g^{-1})g$~\cite{duenez2004random,matsumoto2013weingarten}. Note that this measure is left-$K$-invariant; indeed let $g=kg_0,\ k\in K,\ g_0\in G$, whence
\begin{equation}
\sg(g^{-1})g=\sg((kg_0)^{-1})kg_0 = \sg(g_0)^{-1}\sg(k)^{-1}kg_0= \sg(g_0)^{-1}k^{-1}kg_0= \sg(g_0)^{-1}g_0,
\end{equation}
where we have used that the elements of $K$ are by definition fixed points of $\sg$, and that $\sg$ being a group homomorphism implies $\sg(g^{-1})=\sg(g)^{-1}$.
Our primary occupation in the following appendices will be to in various circumstances integrate the adjoint action of $G/K$ against this measure, i.e. ``twirling'' over $G/K$. As we shall see, the twirl over a symmetric space is somewhat less nice than its analogue over groups, though not entirely without structure. For example (using the notation $\mcl:=\eh$ and $\mct^{(n)}_{ \mce}: \mcl\tn\to\mcl\tn, \mct^{(n)}(A)= \expect_{U\sim \mce}U\tn A (U\ad)\tn$), we have

\begin{restatable}{lem}{equi}\label{lem:equi}
The $n$\textsuperscript{th} order twirl $\mct^{(n)}_{G/K}:\mcl\tn\to\mcl\tn$ is $K$-equivariant.
\end{restatable}
\begin{proof}
Let $A\in\mcl\tn,\ k\in K$; we directly calculate:
\begin{align}
    \left(\mct^{(n)}_{G/K} \circ \Ad_k\tn\right)(A)&=\mct^{(n)}_{G/K} (k\tn A k^{\dagger \otimes n})\\
    &=\int_{v\in G/K} v\tn k\tn A k^{\dagger \otimes n}v^{\dagger \otimes n}\\
    &=\int_{g\in G} (\sg(g)\ad g)\tn k\tn A k^{\dagger \otimes n}((\sg(g)\ad g)\ad)\tn\\
    &=\int_{g\in G} (\sg(g\ad) gk)\tn A (k\ad g\ad \sg(g) )\tn\\
    &=\int_{g\in G} (\sg(kg\ad) g)\tn A ( g\ad \sg(gk\ad) )\tn \label{eq:a6}\\
    &=\int_{g\in G} (\sg(k)\sg(g\ad) g)\tn A ( g\ad \sg(g)\sg(k\ad) )\tn\\
    &=\int_{g\in G} k\tn (\sg(g\ad) g)\tn A ((  \sg(g)\ad g)\ad)\tn k^{\dagger \otimes n}\\
    &=\int_{v\in G/K} k\tn v\tn A v^{\dagger \otimes n} k^{\dagger \otimes n}\\
    &= \left(  \Ad_k\tn \circ \mct^{(n)}_{G/K}\right)(A)
\end{align}
where on various occasions we used that $\sg$ is a group homomorphism, and in Eq.~\eqref{eq:a6} we used the right-invariance of the Haar measure on $G$ to make the change of variables $g\mapsto gk\ad$. Subsequently, we used that $\sg(k)=k\ \forall k\in K$.
\end{proof}
The statement of Lemma~\ref{lem:equi} is much weaker than that which can be made in the group case: $\mct^{(n)}_{G}$ is $G$-\textit{invariant}, and in fact projects onto the
$n$\textsuperscript{th} order commutant of $G$~\cite{mele2023introduction}. We will later see that $\mct^{(n)}_{G/K}$ is not even idempotent, and therefore cannot be a projector onto anything. More helpfully, Lemma~\ref{lem:equi} combined with Schur's lemma~\cite{fulton1991representation} tells us that the  twirl operator decomposes as
\begin{equation}
    \mct^{(n)}_{G/K} = \bigoplus_{\lm} \id_{d_\lm}\ot s_\lm,
\end{equation}
where the sum is over the $K$-irreps in $\mcl\tn$, with the corresponding multiplicity spaces furnishing the only non-trivial action of $\mct^{(n)}_{G/K}$. Recalling, for example, that under the adjoint action of the orthogonal group $\mcl$ decomposes as 
\begin{equation}
    \mcl\ \cong\  {\rm Sym}^2\mch \hspace{0.7mm} \oplus \hspace{0.7mm}\Lambda^2\mch\ \cong\  \mbc  \oplus U \oplus \hspace{0.7mm}\Lambda^2\mch 
\end{equation}
(with $U$ the orthogonal complement of the trivial $\mbo$-representation furnished by $\id$ in ${\rm Sym}^2\mch$), we see that 
\begin{align}
    \mct^{(1)}_{\rm AI}(A) &=  \a (A)\mcp_{\id}(A) + \b(A)\mcp_{U}(A) + \g(A) \mcp_{\Lambda^2\mch}(A)\\
    &= \a(A) \id + \b (A)\left(\frac{A+A^\mst}{2}-\frac{\tr[A]}{d}\id\right)+ \g (A)\left(\frac{A-A^\mst}{2}\right) \label{eq:aibaby}
\end{align}
for scalars $\a(A),\b(A),\g(A)$ that we will determine in Appendix~\ref{sec:ai}. Recalling that the first-order twirl over $G\in\{\mbu(d),\mbo(d),\mbsp(d)\}$ is given simply by $\mct^{(1)}_{G}(A) =(1/d)\tr[A]\id$, Eq.~\eqref{eq:aibaby} foreshadows the increased complexity that we will encounter when integrating over symmetric spaces.\\

Now, it is evident from the above discussion that a $k$\textsuperscript{th}-order integral over $G/K$ may be recast as a $(2k)$\textsuperscript{th}-order integral over $G$, allowing us to in our examples employ the standard Weingarten calculus on the unitary, orthogonal and symplectic groups~\cite{collins2006integration}. For example, recalling from the main text that $\sg_{\rm AI}(A)=A^*$, one sees that
\begin{equation}
    \mct^{(k)}_{\rm AI}(-) = \int_{g\in{\rm AI}} g\tk (-) (g\ad)\tk = \int_{U\in \mbu} (U^\mst U)\tk (-) (U\ad U^*)\tk ;
\end{equation}
this is a technique we will use on multiple occasions throughout the following appendices. 
It was recently shown by Matsumoto~\cite{matsumoto2013weingarten}, however, that one can integrate directly over the symmetric spaces, without resorting to a doubled-order integral on $G$; the rest of this appendix is devoted to (briefly) reviewing the results of that work.\\

Let $G/K$ be a classical compact symmetric space $\mcc$ realised as a matrix ensemble via $V=\sg(g)^{-1}g$ for $g\sim\mu_G$ and an involution $\sg$ on $G$, and recall the tabulation of all such spaces in Table~\ref{tab:symm_classif}. We let $\tilde v = Jv$. The main result of Ref.~\cite{matsumoto2013weingarten} is that moments of products of matrix entries $v_{ij}$ of $V$ can be written as a single sum over the symmetric group (for type A) or, for types B/C/D, the symmetric group mod its \textit{hyperoctahedral subgroup} (see below), with a class–dependent ``delta tensor'' and a Weingarten function\footnote{Explicit expressions for which are derived in Ref.~\cite{matsumoto2013weingarten}} ${\rm Wg}^{\mathcal C}$:
\begin{alignat}{2}
  \expect_{V\sim\,\rm{ AI}}\!\ \prod_{r=1}^k v_{i_{2r-1} i_{2r}}\overline{v_{j_{2r-1} j_{2r}}} &= \sum_{\sigma\in S_{2k}}\ \delta_\sigma(\bm i,\bm j)\ {\rm Wg}^{ \rm{ AI}}(\sigma )  &   \label{eq:a16}\\
  \expect_{V\sim\,\rm{ AII}}\!\ \prod_{r=1}^k \tilde{v}_{i_{2r-1} i_{2r}}\overline{\tilde{v}_{j_{2r-1} j_{2r}}} &= \sum_{\sigma\in S_{2k}}\ \delta_\sigma(\bm i,\bm j)\ {\rm Wg}^{\rm{ AII}}(\sigma )  & \label{eq:a17}\\
    \expect_{V\sim\,\rm{ AIII}}\!\ \prod_{r=1}^k v_{i_{r} j_{r}}  &= \sum_{\sigma\in S_{k}}\ \delta_\sigma(\bm i,\bm j)\ {\rm Wg}^{\rm{ AIII}}(\sigma;p,q)  & \label{eq:a18}\\
  \expect_{V\sim\,\rm{BDI}}\!\ \prod_{r=1}^k v_{i_{2r-1} i_{2r}} &= \sum_{\sigma\in S_{2k}/H_k}\ \Delta^{\rm{BDI}}_\sigma(\bm i)\ {\rm Wg}^{\rm{BDI}}(\sigma;p,q)\quad &\label{eq:a19}\\
  \expect_{V\sim\,\mcc}\!\ \prod_{r=1}^k \tilde v_{i_{2r-1} i_{2r}} &= \sum_{\sigma\in S_{2k}/H_k}\ \Delta^{\mathcal C}_\sigma(\bm i)\ {\rm Wg}^{\mathcal C}(\sigma;\text{params})\,,\qquad &\qquad\mcc\in\{{\rm CI,\, CII,\, DIII}\}\label{eq:a20}
\end{alignat}
This stands in a perhaps somewhat surprising contrast to the classical group case ($G\in\{\mbu,\mbo,\mbsp\}$), where moments are given by a double sum $\sum_{\sigma,\tau}\Delta_\sigma^G(\bm i)\Delta_\tau^G(\bm j){\rm Wg}^{G}(\sigma^{-1}\tau;d)$~\cite{collins2006integration,mele2023introduction} for ``delta-tensors'' $\Delta^G$. In our case the delta-tensors depend on the symmetric space as
\begin{equation}\label{eq:tri}
\Delta^{\mathcal C}_\sigma =
\begin{cases}
\delta_\sigma(\bm i,\bm j)=\prod_{r=1}^{{\rm len}(\bm i)}\delta_{\,i_r,\ j_{\sigma(r)}}, & \mathcal C\in\{{\rm AI},{\rm AII},{\rm AIII}\},\\
\Delta_\sigma(\bm i)=\prod_{r=1}^{{\rm len}(\bm i)/2}\delta_{\,i_{\sigma(2r-1)},\ i_{\sigma(2r)}}, & \mathcal C\in\{{\rm BDI},{\rm DIII}\},\\
\Delta'_\sigma(\bm i)=\prod_{r=1}^{{\rm len}(\bm i)/2}  e_{\,i_{\sigma(2r-1)}}^\mst{J\,e_{\,i_{\sigma(2r)}}}, & \mathcal C\in\{{\rm CI},{\rm CII}\}
\end{cases}
\end{equation}
where $J=\begin{pmatrix}
    0&-\id\\\id&0
\end{pmatrix}$ is the canonical antisymmetric bilinear form and ${\rm len}(i)$ is the number of entries in the vector $\bm i$.
Recalling that the hyperoctahedral group $H_k\le S_{2k} \cong \mbz_2 \wr S_k $ is the subgroup  that permutes the $k$ disjoint pairs $\{(1,2),(3,4),\dots,(2k\!-\!1,2k)\}$ and flips elements inside each pair, we see that the sums in Eqs.~\eqref{eq:a19} and~\eqref{eq:a20} are over sets $\{ \{\sigma(1),\sigma(2)\}, \{\sigma(3), \sigma(4)\},\dots,\{
\sigma(2k-1), \sigma(2k)\}\}$ of pairs of elements where  $\sigma \in S_{2k}$ satisfies
$\sigma(2i-1)< \sigma(2i) \ (1 \le i \le k) $  and $ \sigma(1)< \sigma(3)< \cdots < \sigma(2k-1). $

\section{Classical shadows over symmetric spaces}\label{sec:cs}
\noindent
From the algorithmic description of the  setup in the main text, one sees immediately that the classical shadow ``measurement channel'' $\mcm_{\mce,\mcw}$ (that is, the quantum channel describing the procedure) is given by 
\begin{equation}
    \mcm_{\mce,\mcw}(\rho) =\sum_{w\in\mcw}\expect_{U\sim \mce} \tr \left[U \rho U \ad \Pi_w \right]  U^{\dagger}\Pi_wU,\label{eq:mc1}
\end{equation}
where we employ the notation $\Pi_w\equiv\ketbra{w}$; a little graphical manipulation gives
\begin{align}
    \sum_{w\in\mcw}\expect_{U\sim \mce} \tr \left[U \rho U \ad \Pi_w \right]  U^{\dagger}\Pi_wU 
    &=\sum_{w\in\mcw}\expect_{U\sim \mce}\img[1.65cm]{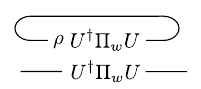}\\
    &=  \tr_1\left[\left(\rho \otimes \id\right) \expect_{U\sim \mce}\left\{ \sum_{w\in\mcw}\left( U^{\dagger}\Pi_wU\right)\hspace{-0.7mm}{}^{\otimes 2}\right\}\right]\\
    &= \tr_1\left[\left(\rho \otimes \id\right) \mct^{(2)}_{\mce} \left( \sum_{w\in\mcw} \Pi_w\ts \right)\right]\,,
\end{align}
where in the last line we have used the invariance of our considered ensembles under the transformation $U\mapsto U\ad$.
We will also occasionally find it useful to work with the equivalent formulation
\begin{equation}
\mcm_{\mce,\mcw} = \expect_{g\sim \mce} {\rm Ad}_{g^\dagger}\circ \mca_\mcw \circ {\rm Ad}_{g}
\end{equation}
where $\mca_\mcw(-)=\sum_{w\in\mcw} \bra{w}(-)\ket{w}\ketbra{w}$ is a dephasing channel with respect to $\mcw$.
Now, were $\mce$ the uniform ensemble over a group $G$, the measurement channel would be manifestly $G$-equivariant; indeed note that for any $h\in G$ an appropriate change  of variables ($g\mapsto gh$) in the integral yields
\begin{align*}
    {\rm Ad}_{h}\circ \mcm_{G,\mcw} = \expect_{g\sim G} {\rm Ad}_{h}\circ {\rm Ad}_{g^\dagger}\circ \mca_\mcw \circ {\rm Ad}_{g} = \expect_{g\sim G} {\rm Ad}_{h g^\dagger}\circ \mca_\mcw \circ {\rm Ad}_{g}&= \expect_{g\sim G} {\rm Ad}_{ g^\dagger}\circ \mca_\mcw \circ {\rm Ad}_{gh}\\&= \expect_{g\sim G} {\rm Ad}_{ g^\dagger}\circ \mca_\mcw \circ {\rm Ad}_{g}\circ {\rm Ad}_{h} 
    = \mcm_{G,\mcw}\circ {\rm Ad}_{h};
\end{align*}
such a channel would therefore have a nice decomposition over $G$-irreps in $\mcl$,
\begin{equation}
    \mcm_{\mu_G,\mcw} = \bigoplus_{\lm} \id_{d_\lm}\ot s_\lm,
\end{equation}
and our first question is whether something similar holds for $\mce=\mu_{G/K}$, the invariant measure on a compact symmetric space.\\ 

We begin by noting that, although $\mct^{(k)}_{G/K}$ is $K$-equivariant (Lemma~\ref{lem:equi}), the measurement channel is \textit{not}; in general, there will be elements $k\in K$ for which $ \mcm\circ {\rm Ad}_{k} \neq {\rm Ad}_{k}\circ\mcm $. Indeed, for $k\in K$ a direct calculation gives
\begin{align}
    (\mcm_{G/K,\mcw} \circ {\rm Ad}_{ k})(\rho) &= \sum_{w\in\mcw}\expect_{U\sim G/K}  \tr[Uk\rho k\ad U\ad \Pi_{w} ] U\ad \Pi_{w} U\\
    &= \sum_{w\in\mcw}\expect_{g\sim G}  \tr[\sg(g\ad)gk\rho k\ad g\ad \sg(g) \Pi_{w} ] g\ad\sg(g) \Pi_{w} \sg(g\ad)g\\
    &= \sum_{w\in\mcw}\expect_{g\sim G}  \tr[\sg(kg\ad)g\rho  g\ad \sg(gk\ad) \Pi_{w} ] kg\ad\sg(gk\ad) \Pi_{w} \sg(kg\ad)gk\ad\\
    &= \sum_{w\in\mcw}\expect_{g\sim G}  \tr[\sg(g\ad)g\rho  g\ad \sg(g)k\ad \Pi_{w}k ] kg\ad\sg(g)k\ad \Pi_{w} k\sg(g\ad)gk\ad\\
    &= \sum_{w\in\mcw}\expect_{g\sim G}  \tr[\sg(g\ad)g\rho  g\ad \sg(g) {\rm Ad}_{k\ad} \left(\Pi_{w}\right) ] kg\ad\sg(g) {\rm Ad}_{k\ad} \left(\Pi_{w}\right)\sg(g\ad)gk\ad\\
    &= k\left(\sum_{w\in\mcw}\expect_{U\sim G/K}  \tr[U\rho  U\ad {\rm Ad}_{k\ad} \left(\Pi_{w}\right) ] U\ad {\rm Ad}_{k\ad} \left(\Pi_{w}\right)U \right)k\ad\\
    &= ({\rm Ad}_{ k}\circ\mcm_{G/K,{\rm Ad}_{k\ad}(\mcw)}  )(\rho)
\end{align}
where ${\rm Ad}_{k\ad}(\mcw)$ denotes the basis with elements given by $\{k\ad\ket{w}\}_{w\in\mcw}$. As different choices of measurement basis generically lead to distinct shadow protocols, we see that in general $\mcm$ fails to commute with ${\rm Ad}_K$. Evidently there is nonetheless a subgroup of  $K$ for which it does; indeed let us say a group $H$ is $\mcw$\textit{-normalising} if, for all $h\in H$ and $\ket{w}\in\mcw$, there exists a $\ket{w'}\in\mcw$ such that $h\ketbra{w}h\ad=\ketbra{w'}$. We will use $ N _\mcw$ to denote the group of all $\mcw$-normalising elements of $G$.  The above calculation combined with Schur's lemma then immediately implies: 
\lemstab*

\noindent
Now, in principle  the decomposition of Eq.~\eqref{eq:symstruct} may not be particularly simple; indeed, the $s_\lm$ are at this point in the story arbitrary dense matrices on the multiplicity spaces of the irrep decomposition of $\eh$, which, a priori, could be very large. Happily, however, we find that in all cases the shadow channel takes a very simple form, given by:

\addtocounter{thm}{-1} 
\begin{restatable}{thm}{thmapp} 
Sampling uniformly from $G/K$ and measuring in the basis $\mcw$ with respect to which the involutions take their standard forms yields
    \begin{equation}
        \mcm_{G/K,\mcw}(\rho) = (1-\a_{G/K})\mcm_{G,\mcw}(\rho) + \b_{G/K}\mca_{\mcw}(\rho) +(\a_{G/K}-\b_{G/K})\left(J\mca_\mcw(\rho)J\ad - \mca_\mcw(\rho J)J\right),\label{eq:thm1app}
    \end{equation}
    where $\mca_{\mcw}$ is a dephasing channel onto $\mcw$ and $J$ is the canonical symplectic form.
    For AI, AII, CI, and DIII, $\a_{G/K}\in\mco(d^{-2})$; for AIII, BDI and CII, $0\le \a_{G/K} \le 1$ is tunable. Finally, $\a_{G/K}=\b_{G/K}$ for $G\neq \mbsp$, and $|\a_{G/K}-\b_{G/K}|\in\mco(1/d)$ for $G=\mbsp$. 
\end{restatable}
\begin{proof}
Let us begin by rewriting Eq.~\eqref{eq:mc1} in terms of its matrix elements (with respect to the basis of $\mcw$):
\begin{align}
 \mcm_{\mce,\mcw}(\rho) &=\sum_{w\in\mcw}\expect_{V\sim \mce} \tr \left[V \rho V \ad \Pi_w \right] V^{\dagger}\Pi_wV\\
 &=\sum_{w\in\mcw}\sum_{i,j,a,b} \rho_{a,b}\expect_{V\sim \mce}[\overline{v_{w,b}v_{w,i}}v_{w,a}v_{w,j}] \ketbra{i}{j}\label{eq:b15}
\end{align}
We are in a position to directly employ the results of Ref.~\cite{matsumoto2013weingarten}, which we can do on a case-by-case basis according to the pentachotomy of Eqs.~\eqref{eq:a16},~\eqref{eq:a17},~\eqref{eq:a18},~\eqref{eq:a19} and~\eqref{eq:a20}. For example, consider the case that $\mce$ is the uniform measure on AI, whence by Eqs.~\eqref{eq:a16} and~\eqref{eq:tri} we have
\begin{align}
   \sum_{w\in\mcw} \expect_{V\sim {\rm AI}}[\overline{v_{w,b}v_{w,i}}v_{w,a}v_{w,j}] &= \sum_{w\in\mcw}\sum_{\sg\in S_4} \delta_\sg((w,b,w,i),(w,a,w,j)){\rm Wg}^{\rm AI}(\sg)\label{eq:b16}\\
    &= \lm_{\rm AI} \delta_{a,b}\delta_{i,j}+\lm_{\rm AI}' \delta_{a,i}\delta_{b,j}+\lm_{\rm AI}'' \delta_{a,b,i,j}
\end{align}
for some coefficients $\lm_{\rm AI},\lm_{\rm AI}',\lm_{\rm AI}''$; by the manifest symmetry of the LHS of Eq.~\eqref{eq:b16} under the exchange $b \leftrightarrow i$, we further see that $\lm_{\rm AI}=\lm_{\rm AI}'$.  Substituting this into Eq.~\eqref{eq:b15} yields
\begin{align}
    \mcm_{{\rm AI},\mcw}(\rho)  &= \lm_{\rm AI} (\tr[\rho]\id +  \rho) + \lm_{\rm AI}'' \sum_i \rho_{i,i}\ketbra{i}{i} \\
    &=  \lm_{\rm AI} (d+1)\left( \frac{\tr[\rho]\id +  \rho}{d+1}\right) + \lm_{\rm AI}'' \mca_\mcw(\rho) \\
    &= \left(1-\a_{\rm AI}\right) \mcm_{\mbu,\mcw}(\rho)   +  \a_{\rm AI} \mca_\mcw(\rho)\label{AI_form} 
\end{align}
In the final line we have recalled that the measurement channel in the unitary case is exactly a depolarising channel of strength $p=d/(d+1)$,  made the definition $\a_{\rm AI}:=1-(d+1)\lm_{\rm AI}$, and noticed that by trace preservation $\a_{\rm AI}=\lm_{\rm AI}''$.  The case of AIII proceeds from Eq.~\eqref{eq:a18} by identical logic and the observation that  (up to the harmless action of a diagonal unitary matrix) the elements of AIII are self-adjoint. Now, to fix the remaining coefficient we need to make explicit use of the known form of the Weingarten coefficients; a direct calculation via computer algebra yields 
\begin{align}
    \alpha_{\rm AI}&=\frac{2}{d(d+3)}\\
    \alpha_{\rm AIII}&=\frac{ s^4  + 2 s^2 (d-2)  + d^2}{d^2(d-1)(d+3)},\label{eq:aiiialpha}
\end{align}
where $s=p-q$ is the signature of $I_{p,q}$. Note that $\alpha_{\rm AI}\in\mco(d^{-2})$ but that $\alpha_{\rm AIII}$ can be $\mco(1)$ (for $s\sim d$).  The details behind such a calculation are (in the AI case) given in Appendix~\ref{sec:ai}; indeed we will see several ways to obtain the answer. 
\\

The sole complication introduced in the case of AII is that our expression Eq.~\eqref{eq:a17} is in terms of the entries of the random matrix $\tilde V = JV$, not $V$ itself (where $V$ is the matrix drawn uniformly from AII). This, however, presents no fundamental difficulty; indeed we have

\begin{align}
 \mcm_{{\rm AII},\mcw}(\rho) &=\sum_{w\in\mcw}\expect_{V\sim {\rm AII}} \tr \left[V \rho V \ad \Pi_w \right] V^{\dagger}\Pi_wV\\
 &=\sum_{w\in\mcw}\expect_{V\sim {\rm AII}} \tr \left[\tilde V \rho  \tilde  V \ad J \Pi_w J^{\mst}\right] \tilde  V^{\dagger} J \Pi_w  J^{\mst}\tilde  V\\
 &=\sum_{w\in\mcw}\expect_{V\sim {\rm AII}} \tr \left[\tilde V \rho  \tilde  V \ad \Pi_w  \right] \tilde  V^{\dagger}  \Pi_w  \tilde  V\\
 &=\sum_{w\in\mcw}\sum_{i,j,a,b} \rho_{a,b}\expect_{V\sim {\rm AII}}[\overline{\tilde v_{w,b} \tilde v_{w,i}}\tilde v_{w,a}\tilde v_{w,j}] \ketbra{i}{j} 
\end{align}
where we have used that $J$ acts as a signed permutation on $\mcw$, so that the sum over $w$ is effectively merely rearranged. At this point we can employ Eq.~\eqref{eq:a17} and subsequently proceed to the desired conclusion exactly as in the cases of AI and AIII. The specific value of the non-symmetry-fixed coefficient is calculated by direct computer algebra (using the methods of Appendix~\ref{sec:ai}) to be 
\begin{equation}
    \a_{\rm AII} = \frac{2}{d(d-1)}\,.
\end{equation}

Next let us turn (jointly) to the cases of BDI and DIII. Their equivalence at the structural level follows from their shared behaviour in Eq.~\eqref{eq:tri}, with their distinction at the levels of Eq.~\eqref{eq:a19} and Eq.~\eqref{eq:a20} being solely in the introduction of the ``tilded variables'', which as we have  seen is for our purposes immaterial. To be concrete, let us consider specifically BDI, whence Eqs.~\eqref{eq:a19} and~\eqref{eq:tri} yield
\begin{align}
   \sum_{w\in\mcw} \expect_{V\sim {\rm BDI}}[{v_{b,w}v_{w,a}}v_{i,w}v_{j,w}] &= \sum_{w\in\mcw}\sum_{\sg\in S_8 / H_4} \Delta_\sg((b,w,w,a,i,w,w,j)){\rm Wg}^{\rm BDI}(\sg)\label{eq:bdithisone} \\
    &= \lm_{\rm BDI} \delta_{a,b}\delta_{i,j}+\lm_{\rm BDI}' \delta_{a,i}\delta_{b,j}+\lm_{\rm BDI}'' \delta_{a,j}\delta_{b,i}+\lm_{\rm BDI}''' \delta_{a,b,i,j}
\end{align}
where we note with interest the appearance of a term missing from the unitary case of Eq.~\eqref{eq:b16}. By the manifest symmetry of Eq.~\ref{eq:bdithisone} under the permuting of $a,b,i$ and $j$ we obtain 
\begin{align}
    \mcm_{{\rm BDI},\mcw}(\rho)  &= \lm_{\rm BDI} (\tr[\rho]\id +  \rho+\rho^\mst) + \lm_{\rm BDI}''' \mca_\mcw (\rho)  \\
    &=  \left(1-\a_{\rm BDI}\right) \mcm_{\mbo,\mcw} (\rho)  +  \a_{\rm BDI} \mca_\mcw(\rho) 
\end{align}
and, \textit{mutatis mutandis}, the same expression for DIII. As before, we can fix the precise values of the coefficients using brute force computer algebra, finding
\begin{align}
  \a_{\rm BDI}&=  \frac{  s^{4}  + (6d-4)s^2  + 3 d(d - 2)}{d(d-1)(d+1)(d+6)}\\
  \a_{\rm DIII}&=  \frac{3}{d^2-1},
\end{align}
where  $s=p-q$ reprises its role as the signature of $I_{p,q}$.

Finally we turn to the (structurally identical) cases of CI and CII, taking for notational concreteness the case of CI. Up to the action of an (irrelevant for our purposes, as the phases cancelled) diagonal matrix, the random variable $V\sim \mu_{\rm CI}$ is self-adjoint, so that the relevant expectation value for the shadow channel becomes

\begin{align}
   \sum_{w\in\mcw} \expect_{V\sim {\rm CI}}[{v_{b,w}\ad v_{w,a}}v_{i,w}\ad v_{w,j}] 
   &=\sum_{w\in\mcw} \expect_{V\sim {\rm CI}}[{v_{b,w}v_{w,a}}v_{i,w}v_{w,j}]\label{eq:ciearly} \\
   &=\sum_{w\in\mcw} \expect_{V\sim {\rm CI}}[\pm{\tilde v_{b^\sharp,w}\tilde v_{w^\sharp,a}}\tilde v_{i^\sharp,w}\tilde v_{w^\sharp,j}] \\
   &= \sum_{w\in\mcw}\pm\sum_{\sg\in S_8 / H_4} \Delta_\sg'((b^\sharp,w,w^\sharp,a,i^\sharp,w,w^\sharp,j)){\rm Wg}^{\rm CI}(\sg)  \\
    &= \lm_{\rm CI}^{(1)} \delta_{a,b}\delta_{i,j}+\lm_{\rm CI}^{(2)} \delta_{a,i}\delta_{b,j}+\lm_{\rm CI}^{(3)} J_{a,j}J_{b,i}+\lm_{\rm CI}^{(4)} \delta_{a,b,i,j}+\lm_{\rm CI}^{(5)} \delta_{a,b}J_{a,i}J_{a,j}+\lm_{\rm CI}^{(6)} \delta_{a,i}J_{a,b}J_{a,j}\label{eq:cifun}
\end{align}
where we have used Eq.~\eqref{eq:tri} and a superscript $\sharp$ to denote the symplectic partner index (that is, the unique index such that $J_{a,a^\sharp}\neq 0$). Now, by the symmetry of  Eq.~\eqref{eq:ciearly} under the exchange $b\leftrightarrow i$, we clearly have $\lm_{\rm CI}^{(1)} =\lm_{\rm CI}^{(2)} $, $\lm_{\rm CI}^{(5)} =\lm_{\rm CI}^{(6)} $ and $\lm_{\rm CI}^{(3)} =0$, whence
\begin{align}
    \mcm_{{\rm CI},\mcw}(\rho)  &= \lm_{\rm CI}^{(1)}(\tr[\rho]\id + \rho) + \lm_{\rm CI}^{(4)}\mca_\mcw(\rho) + \lm_{\rm CI}^{(5)}(J \mca_\mcw(\rho) J\ad - \mca_\mcw(\rho J)J)\\
    &=(1-\a_{\rm CI})(\mcm_{\mbsp,\mcw}(\rho)) + \b_{\rm CI}\mca_\mcw(\rho) + (\a_{\rm CI}-\b_{\rm CI})(J \mca_\mcw(\rho) J\ad - \mca_\mcw(\rho J)J)
\end{align}
where the introduction of the variables $\a_{\rm CI}$ and $\b_{\rm CI}$ to replace $\lm_{\rm CI}^{(1)},\,\lm_{\rm CI}^{(4)},$ and $\lm_{\rm CI}^{(5)}$ is as usual possible due to trace preservation, and we have used that the $\mbsp$ shadows channel is simply a depolarising channel of strength $p=d/(d+1)$~\cite{west2024random}. The derivation for the CII case is identical. For the specific values of the coefficients we again turn to brute force computer algebra, finding
\begin{align}
    \a_{\rm CI} &= \frac{ 3}{(2n-1)(2n+3)}\\
    \b_{\rm CI} &= \frac{6n+1}{(2n-1)(2n+1)(2n+3)}\\
    \a_{\rm CII} &= \frac{ 4 s^4- 10 s^2+ 3 n^2 + 3 n  }{n(n-1)(2n-1) (2n+3)}\\
    \b_{\rm CII} &= \frac{3 n (n + 1) (2 n^2 + n + 1) + 4 s^2 ((2 n^2 + n + 1) s^2 + 2 n^3 - 5 n^2 - 6 n - 1)}{n(n-1)(n+1)(2n-1)(2n+1)(2n+3)}
\end{align}
where we have introduced $n=d/2$.
We note that, similarly to previous cases, $\a_{\rm CI}\in\mco(d^{-2})$, but $\a_{\rm CII}\in\mco(1)$ for $s\sim n$, and that 
\begin{equation}
    |\a_{\rm CII}-\b_{\rm CII}| = \frac{|2 (s^2 - n^2) (4 ns^2 - 3 n - 3)|}{n(n-1)(n+1)(2n-1)(2n+1)(2n+3)}\in\mco(1/n)
\end{equation}
for all allowable values of $s$.

\end{proof}

\subsection{Sample complexities}
Next let us turn our attention to the question of the \textit{variance} of the classical shadow estimators. By the general relation $\Var [\hat{o}] = \mbe[\hat{o}^2]-\mbe[\hat{o}]^2\leq \mbe[\hat{o}^2]$ and the readily seen Hermiticity of $\mcm$ with respect to the Hilbert-Schmidt inner product, we obtain the commonly used bound
\begin{align}
    \Var_{\mce,\mcw} [\hat{o}] &\leq  \sum_w\int_{V\sim \mce}\tr\left[\rho V^\dagger\Pi_wV\right]\tr\left[O \mcm^{-1}\left(V^\dagger\Pi_wV\right)\right] ^2  \\
     &= \sum_w\int_{V\sim \mce}   \tr\left[\rho V^\dagger\Pi_wV\right]\tr\left[\mcm^{-1}\left(O\right) V^\dagger\Pi_wV\right] ^2  \\
     &= \sum_w\int_{V\sim \mce}    \tr\left[\left(\rho\otimes \mcm^{-1}\left(O\right)\otimes\mcm^{-1}\left(O\right)\right)\left(V^{\dagger\otimes 3}\Pi_w^{\otimes 3}V^{\otimes 3}\right) \right] \\
     &=     \tr\left[\left(\rho\otimes \mcm^{-1}\left(O\right)\otimes\mcm^{-1}\left(O\right)\right)\mct^{(3)}_\mce\left(\sum_w\Pi_w^{\otimes 3}\right) \right]
     \end{align}
where we have again used the invariance of our ensembles under the transformation $V\mapsto V\ad$. 
In fact, one can equivalently~\cite{huang2020predicting} consider the traceless $O_0:=O-\tr[O]\id/d$; this turns out to simplify things a little.

\subsubsection{AIII}
\noindent
We turn to bounding the variance of the AIII protocol. Writing $X=\mcm^{-1}(O_0)$ and $D=\mca_\mcw(X)$ we have
\begin{align}
    \mbe[\hat{o}_0^2] &= \sum_w\int_{V\sim \mu_{\rm AIII}}\tr\left[\rho V^\dagger\Pi_wV\right]\tr\left[O_{ 0} \mcm^{-1}\left(V^\dagger\Pi_wV\right)\right] ^2  \\
     &= \sum_w\int_{V\sim \mu_{\rm AIII}}    \tr\left[\left(\rho\otimes  X \otimes X \right)\left(V^{\dagger\otimes 3}\Pi_w^{\otimes 3}V^{\otimes 3}\right) \right] \\
     &=  \tr\Bigg[\left(\rho\otimes  X \otimes X \right)\Bigg(\sum_{abcefg}\ketbra{abc}{efg}\Big( c_1(\delta_{a e} \delta_{b f} \delta_{c g}+\delta_{a e} \delta_{b g} \delta_{c f}+\nonumber\\
     &\hspace{10mm}+\delta_{a f} \delta_{b e} \delta_{c g}+\delta_{a f} \delta_{b g} \delta_{c e}+\delta_{a g} \delta_{b e} \delta_{c f}+\delta_{a g} \delta_{b f} \delta_{c e} )+ c_2 (\delta_{a bef} \delta_{c g}  +  \delta_{a beg} \delta_{c f}+ \delta_{a bfg}  \delta_{c e}  \nonumber\\
     &\hspace{10mm}  +  \delta_{a cef}  \delta_{b g}    +  \delta_{a ceg} \delta_{b f} + \delta_{a cfg} \delta_{b e}  +  \delta_{a e} \delta_{b cfg} + \delta_{a f} \delta_{b ceg}  +  \delta_{a g} \delta_{b cef} ) + c_3\delta_{a bcefg} \Big)\Bigg) \Bigg] \label{eq:aiii16} \\
     &= c_1\Big(\tr[X^2]+2\tr(\rho X^2)\Big)+c_2\Big(2\tr(\mca_\mcw(\rho)X^2)+2\tr[\rho\{D,X\}]+\tr[D^2]\Big) + c_3\,\tr[\rho D^2]
     \end{align}
where
\begin{align}
    c_1 &= \frac{(s^2-d^2)(s^2-(d+2)^2)(d^3+8d^2+2ds^2+7d+6s^2-36)}{d^2(d-1) (d+1)^2(d+2)(d+3)(d+4)(d+5)}\\
    c_2&=-\frac{(s^2-d^2)(d^3 + 2d^2s^2 + 7d^2 + ds^4 + 2ds^2 + 20d +3s^4 - 20s^2 + 32)}{d^2(d-1)(d+1)^2(d+3)(d+4)(d+5)}\\
    c_3&= \frac{(s^2-1)(s^2+3d-2s)(s^2+3d+2s)}{d^2(d-1) (d+1)^2(d+5)}
\end{align}
were obtained by brute force symbolic Weingarten calculus. 
Splitting $O_0$ into diagonal and off-diagonal parts, $O_D = \mca_\mcw(O_0),\, O_O = O_0-O_D$, and using the AIII shadow-channel eigenvalues (which are readily obtained from Eq.~\eqref{eq:aiiialpha})
\begin{align}
    \lm_O  &= \frac{-s^4   - 2s^2(d-2)   + d^2(d^2 + 2 d - 4)}{d^2(d-1)(d+1)(d+3)}\\
    \lm_D  &= \frac{s^4   + 2s^2(d-2)   + d(d^2 + 3 d - 3)}{d(d-1)(d+1)(d+3)}   \,,
\end{align}
so that $X=\lambda_D^{-1}O_D+\lambda_O^{-1}O_O$, our expression becomes
\begin{align}
\expect[\hat{o}_0^2]&=\lambda_D^{-2}(c_1+c_2)\tr(O_D^2)+\lambda_O^{-2}c_1\tr(O_O^2)+\lambda_D^{-2}(2c_1+6c_2+c_3)\tr(\rho O_D^2)\\
&\hspace{5mm}+2\lambda_O^{-2}\Big(c_1\tr(\rho O_O^2)+c_2\tr(\mca_\mcw(\rho)O_O^2)\Big)+2\lambda_D^{-1}\lambda_O^{-1}(c_1+c_2)\tr\!\big(\rho\{O_D,O_O\}\big).
\end{align}
Let us consider two asymptotic regimes. First, when $s\in\mco(1)$ we have $c_1\sim d^{-2},\, c_2\sim d^{-3},\, c_3\sim d^{-4},\, \lambda_O\sim\lambda_D\sim d^{-1}$, whence the leading order behaviour is
\begin{equation}\label{eq:aiiilils}
    \expect[\hat{o}^2_0] \underset{s\in\mco(1)}{\sim} \tr[O_0^2] + 2\tr[\rho O_0^2] + \mco(d^{-1}),
\end{equation}
approximately reproducing that of  unitary shadows~\cite{huang2020predicting}. When $s=cd$ for a fixed non-zero $c\in(-1,1)$, on the other hand, we find
\begin{equation}
    c_1\sim \frac{(1-c^2)^2(2c^2+1)}{d^2},\, c_2\sim \frac{c^4(1-c^2)}{d},\, c_3\sim c^{6},\, \lambda_O\sim\frac{1-c^4}{d},\,\lambda_D\sim c^{4}\,,
\end{equation}
which can change the scaling significantly. Intuitively, this range  should strongly favour the estimation of observables concentrated on the diagonal, which we can see; indeed in the extreme case $O_O=0$ we have
\begin{equation}
    \expect[\hat{o}^2_0] \underset{c\in\mco(1)}{\sim} \frac{(1-c^2)\tr[O_D^2]}{c^4d} + \frac{\tr[\rho O_D^2]}{c^2} \le c^{-4} d^{-1}\|O_D\|_2^2 +c^{-2}\|O_D\|_\infty^2,
\end{equation}
where we have used Holder's inequality. As the (squared) spectral norm is generically suppressed by a factor of $d$ compared to the (squared) 2-norm, this constitutes a significant decrease in the variance from Eq.~\eqref{eq:aiiilils}.

\subsubsection{BDI}
\noindent
To bound the variance of the BDI protocol we can tell a similar story. In the analogous equation to Eq.~\eqref{eq:aiii16}, the combination of Eqs.~\eqref{eq:a19} and~\eqref{eq:tri} leads to 31 possible delta contractions (as opposed to the 16 of Eq.~\eqref{eq:aiii16}), but conceptually there is no real difference. 
With $X=\mcm^{-1}(O_0)$, $D=\mca_\mcw(X)$, and $\widetilde Y = (Y+Y^\mst)/2$,  carrying out the symbolic Weingarten calculus straightforwardly yields
\begin{align}
    \mbe[\hat{o}_0^2] &= \sum_w\int_{V\sim \mu_{\rm BDI}}\tr\left[\rho V^\dagger\Pi_wV\right]\tr\left[O_{ 0} \mcm^{-1}\left(V^\dagger\Pi_wV\right)\right] ^2  \\
     &= \sum_w\int_{V\sim \mu_{\rm BDI}}    \tr\left[\left(\rho\otimes  X \otimes X \right)\left(V^{\dagger\otimes 3}\Pi_w^{\otimes 3}V^{\otimes 3}\right) \right] \\
     &= c_1\Big(   2\tr\big[\widetilde X^2]  +8\tr\big[\widetilde\rho \widetilde X^2\big]\Big)+c_2\Big(4\tr \big[\mathcal{A}_{\mathcal W}(\rho) \widetilde X^2\big]+4\tr \big[\widetilde\rho \{D,\widetilde X\}\big]+\tr[D^2]\Big) + c_3 \tr[\widetilde\rho D^2]     
\end{align}
where
\begin{align}
    c_1 &= \frac{
(d-s)(d+s)(d-s+4)(d+s+4)\big(d^3+19d^2+2ds^2+82d+12s^2-48\big)}{d(d-1)(d+1)(d+2)(d+3)(d+4)(d+6)(d+8)(d+10)} \\
    c_2&=\frac{(d-s)(d+s)\big(3d^3+6d^2s^2+24d^2+ds^4+44ds^2-12d+6s^4-96\big)}{d(d-1)(d+1)(d+2)(d+3)(d+6)(d+8)(d+10)} \\
    c_3&= \frac{15d^3+45d^2s^2+15ds^4+30ds^2-60d+s^6+10s^4-56s^2}{d(d-1)(d+1)(d+2)(d+3)(d+10)}\,.
\end{align}
We may continue to proceed as in the AIII case, splitting $\widetilde O_0$ into diagonal and off-diagonal parts, $\widetilde O_D = \mca_\mcw(\widetilde O_0),\, \widetilde O_O = \widetilde O_0-\widetilde O_D$, and using the similarly readily obtained BDI shadow-channel eigenvalues,
\begin{align}
    \lm_O^{\rm BDI}  &= \frac{2 (d^2-s^2)(d^2+6d+s^2-4)}{d(d-1)(d+1)(d+2)(d+6)}\\
    \lm_D^{\rm BDI}  &= \frac{s^4-4s^2+2d^3+15d^2-12+d(6s^2-8)}{(d-1)(d+1)(d+2)(d+6)}  \,,
\end{align}
we obtain
\begin{align}
\mathbb E[\hat o_0^2]&=\lambda_D^{-2}(2c_1+c_2)\tr(O_D^2)+2\lambda_O^{-2}c_1\tr\big(\widetilde O_O^2\big)+\lambda_D^{-2}(8c_1+12c_2+c_3)\tr(\rho O_D^2)
\nonumber\\
&\hspace{5mm}+4\lambda_O^{-2}\Big(2c_1\tr\big(\rho \widetilde O_O^2\big)+c_2\tr\big (\mathcal A_{\mathcal W}(\rho)\widetilde O_O^2\big)\Big)+4\lambda_D^{-1}\lambda_O^{-1}(2c_1+c_2)\tr \big(\rho\{O_D,\widetilde O_O\}\big).
\end{align}
We again consider two asymptotic regimes. First, when $s\in\mco(1)$ we have $c_1\sim d^{-2},\, c_2\sim d^{-3},\, c_3\sim d^{-4},\, \lambda_O\sim\lambda_D\sim d^{-1}$, whence the leading order behaviour is
\begin{equation}
    \expect[\hat{o}^2_0] \underset{s\in\mco(1)}{\sim} \frac 12 \tr[\widetilde O_0^2] + 2\tr[\rho \widetilde  O_0^2] + \mco(d^{-1}),
\end{equation}
approximately reproducing that of orthogonal shadows~\cite{west2025real}. When $s=cd$ for fixed non-zero $c\in(-1,1)$, on the other hand, we find
\begin{equation}
    c_1\sim \frac{(1-c^2)^2(2c^2+1)}{d^2},\, c_2\sim \frac{c^4(1-c^2)}{d},\, c_3\sim c^{6},\, \lambda_O\sim\frac{2(1-c^4)}{d},\,\lambda_D\sim c^{4}\,,
\end{equation}
almost exactly as in the AIII case. When $O_O=0$ we have
\begin{equation}
    \expect[\hat{o}^2_0] \underset{c\in\mco(1)}{\sim} \frac{(1-c^2)\tr[O_D^2]}{dc^4} +  \frac{\tr[\rho O_D^2]}{c^2} \le c^{-4}d^{-1} \|O_D\|_2^2 + c^{-2}\|O_D\|_\infty^2,
\end{equation}
as in the case of AIII.

\section{The AI case in detail}\label{sec:ai}
We now turn to filling in some more lower-level details of a channel calculation, focusing on the symmetric space AI. 
As discussed in Appendix~\ref{sec:ss}, there are (at least) two ways to integrate over a compact symmetric space $G/K$: (i) directly, using the symmetric space Weingarten calculus of Matsumoto~\cite{matsumoto2013weingarten}, or (ii) indirectly, by transforming to an integral (of doubled order) on the ambient group $G$. Let us begin by illustrating the latter approach in a simple setting:
\begin{restatable}{lem}{ai1}[Lemma 3 of Ref.~\cite{west2026average}]\label{lem:ai1}
For any $A\in\endh$,
\begin{equation}
    \expect_{V\sim \mu_{\rm AI}}VAV^\dagger =  \frac{\tr[A]\id + A^\mst}{d+1}
\end{equation}
\end{restatable}

\begin{proof}
We use that $V\sim {\rm AI}$ is identically distributed to the product $\sg_{\rm AI}(U)\ad U=(U^*)\ad U =U^\mst U$, where $U\sim \mbu$, and indeed (by the invariance of the Haar measure on the unitary group under transposition) to \(UU^T\). Making this substitution  yields (for a review of the graphical notation see e.g. Ref.~\cite{mele2023introduction})
\begin{align}
        \expect_{V\sim {\rm AI}}VAV^\dagger &=  \expect_{U\sim \mbu}\left[U U^\mst A (U U^\mst)\ad\right] \\
        &=\expect_{U\sim \mbu}\ \ \img[1.6cm]{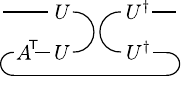}\\
         &= \tr_2\big[(\id\otimes A^T)\expect_{U\sim \mbu}\big[U^{\otimes 2}(d\ketbra{\Phi})(U\ad)^{\otimes 2} \big]\big]\label{eq:vec_approach}\\
         &= \frac{1}{d+1}\tr_2\big[(\id\otimes A^T) (\id+\mbs)\big]\\
        &=\frac{\tr[A]\id + A^\mst}{d+1}
    \end{align}
where the twirl over the unitary group of   \(\ketbra{\Phi}\), with \(\ket{\Phi}\) denoting the Bell state, is readily taken by  a simple application of the Weingarten calculus.
\end{proof}

Lemma~\ref{lem:ai1} also serves to show that, as previously claimed (and in contrast to the familiar case of twirling over a group) $\mct^{(1)}_{\rm AI}$ is \textit{not} a projector, as its eigenspectrum evidently contains $\pm 1/(d+1) \notin \{0,1\} $. Now, as we have previously seen, in order to characterise the classical shadow protocol given by a pair $(\mce,\mcw)$ we need to evaluate the twirls $\mct^{(k)}_{\mce} \left( \sum_{w\in\mcw} \Pi_w\tk \right)$ for $k=2,3$. In the present case, this will be facilitated by the following useful lemma:

\begin{restatable}{lem}{aig}\label{lem:aig}
\begin{equation}\label{eq:aigeneral}
    \int_{V\sim\mu_{\rm AI}}(V\Pi_wV^\dagger)\tk = \sum_{\substack{ \boldsymbol{a}\boldsymbol{b}\boldsymbol{i}\boldsymbol{o} }}  \int_{U\sim\hspace{0.5mm}\mbu} \bra{\boldsymbol{a}\boldsymbol{a}} U^{\ot 2k}
    \left(\Pi_w\tk \ot \ketbra{\boldsymbol{i}}{\boldsymbol{o}}\right) (U\ad)^{\ot 2k} \ket{\boldsymbol{b}\boldsymbol{b}}\ketbra{\boldsymbol{i}}{\boldsymbol{o}}
\end{equation}
where $\boldsymbol{a},\boldsymbol{b},\boldsymbol{i},\boldsymbol{o}\in \mbz_d^k$.
\end{restatable}
\begin{proof}
This is readily seen in the graphical notation~\cite{mele2023introduction}; e.g. for $k=2$ we have

\begin{align}
\img[1.cm]{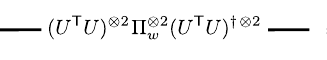}&=\img[1.1cm]{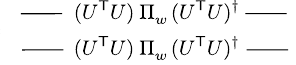}\\
&=\img[1.7cm]{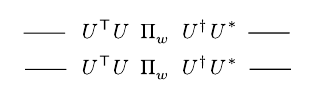}\\
&=\img[3.7cm]{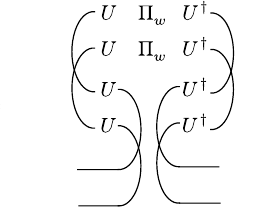}
\end{align}
\noindent
whence the generalisation to arbitrary $k$ is  fairly clear.
\end{proof}

\noindent
At this point, we can reduce the evaluation of the AI shadow channel to a
direct integral over the unitary group: 
\begin{align}
\mcm_{{\rm AI},\mcw}(\rho)&=\sum_{w\in\mcw}  \tr_1\left[\left(\rho \otimes \id\right) \int_{V\sim \rm AI} V^{\otimes 2}\Pi_w^{\otimes 2}V^{\dagger\otimes 2}\right]\\
&=\sum_{w\in\mcw}  \tr_1\left[\left(\rho \otimes \id\right) \int_{U\sim\hspace{0.5mm}\mbu(d)} (U^\mst U)^{\otimes 2}\Pi_w^{\otimes 2}(U^\mst U)^{\dagger\otimes 2}\right]\\
&=\sum_{\substack{ \a \b \ga \d \\ ijkl}}\sum_{w}  \tr_1\left[\left(\rho \otimes \id\right) \int_{U\sim\hspace{0.5mm}\mbu(d)} \bra{\a\b\a\b} U^{\ot 4} \ketbra{wwij}{wwkl} (U\ad)^{\ot 4} \ket{\ga\d\ga\d}\ketbra{ij}{kl} \right]\label{eq:b23}\\
&=\sum_{\substack{ \a \b \ga \d \\ ijkl}}\sum_{w}  \tr_1\left[\left(\rho \otimes \id\right) \sum_{\sg,\tau\in S_4 }{\rm Wg}_d(\tau^{-1}\sg) \braket{\a\b\a\b|\tau}{\ga\d\ga\d}  \sbraket{wwkl|\sg^{-1}}{wwij}  \ketbra{ij}{kl} \right]\\
&=\sum_{\substack{ ijkl}}  \tr_1\left[\left(\rho \otimes \id\right) \left[\frac{(d^{2}+3d-2) (\delta_{ik} \delta_{jl} + \delta_{il} \delta_{jk}) + (2 d+2) \delta_{ij} \delta_{ik} \delta_{il}   }{d (d+1)(d+3)}\right] \ketbra{ij}{kl} \right]\\
&=\tr_1\left[\left(\rho \otimes \id\right) \left[\frac{(d^{2}+3d-2) (\id+\mbs)   }{d (d+1)(d+3)}\right]  \right]+\frac{ (2 d+2)\mca_\mcw(\rho)  }{d (d+1)(d+3)} \\
&=\frac{(d^{2}+3d-2)(\tr[\rho]\id+\rho)}{d(d+1)(d+3)}+\frac{ 2\mca_\mcw(\rho)  }{d (d+3)}\\
&=\left(\mcp_{\id} + \frac{d+5}{(d+1)(d+3)}\mcp_{\rm traceless-diag} + \frac{d^2+3d-2}{d(d+1)(d+3)}\mcp_{\rm off-diag} \right)(\rho)\label{eq:ai_res}\\
&=\left(1-\frac{2}{d(d+3)}\right)\mcd_{d/(d+1)}(\rho)+\frac{ 2\mca_\mcw(\rho)  }{d (d+3)}
\end{align}
where to obtain Eq.~\eqref{eq:b23} we have used Lemma~\ref{lem:aig}, $\mcd_p$ is a depolarising channel of strength $p$, and $\mca_\mcw$ is as usual a completely dephasing channel, $\mca_\mcw(\ketbra{a}{b})=\delta_{a,b}\ketbra{a}$. \\

Alternately, we can note that by Theorem~\ref{thm:thmmain} it remains only to determine the coefficient $\alpha_{\text{AI}}$ appearing in the AI shadow channel in Eq.~\eqref{AI_form}. For this purpose, it suffices to evaluate the channel on a single operator, which we take to be $E_{11}=\ketbra{1}$ and consider the $(1,1)$-matrix element of the output:
\begin{align}
\langle 1|\mcm_{{\rm AI},\mcw}(E_{11})|1\rangle&= \sum_{w\in \mcw}\expect_{V\sim \rm AI}|\langle 1|V|w\rangle|^4\\
&=\sum_{w\in \mcw}\expect_{V\sim \rm AI}|V_{1w}|^4\\
&=\expect_{V\sim \rm AI}|V_{11}|^4+(d-1)\expect_{V\sim \rm AI}|V_{12}|^4,
\end{align}
where in the last step we used permutation invariance of the AI ensemble, which implies that all off-diagonal entries (here $w\neq 1$) in a fixed row are identically distributed. Using Matsumoto’s formula for AI~\cite{matsumoto2013weingarten}, we obtain
\begin{align}
\expect_{V\sim \rm AI}|V_{11}|^4=\frac{8}{(d+1)(d+3)},\qquad  \expect_{V\sim \rm AI}|V_{12}|^4=\frac{2}{d(d+3)}.
\end{align}
Substituting these values into the left-hand side of Eq.~\eqref{AI_form}, and using $\langle1|\mcm_{\mbu,\mcw}(E_{11})|1\rangle=2/(d+1)$ and $\langle 1|\mca_\mcw(E_{11})|1\rangle=1$ for the right-hand side, the comparison of both sides yields
\begin{align}
    \alpha_{\rm AI}=\frac{2}{d(d+3)}.
\end{align}
Hence, we finally obtain 
\begin{align}
\mcm_{{\rm AI},\mcw}(\rho)=\left(1-\frac{2}{d(d+3)}\right)\mcm_{\mbu,\mcw}(\rho)+\frac{ 2\mca_\mcw(\rho)  }{d (d+3)}\,.
\end{align}

Finally, we remark that the expression~\eqref{eq:ai_res} for the shadow channel is compatible with the result of Lemma~\ref{lem:stab}. To see this, note that  $H:= \mbo(d)\cap  N _\mcw \cong \mbz_2 \wr S_d $ is the group of ``signed permutation matrices'' (i.e. those matrices that (when expressed  with respect to $\mcw$) are identically zero except for  exactly one $\pm1$ in each row/column).  By Lemma~\ref{lem:stab}, we know that $\mcm_{\rm AI}$ is equivariant with respect to $H$, and is therefore block-diagonal with respect to the decomposition of $\mcl$ into pairwise-non-isomorphic irreps under ${\rm Ad}_H$, which one sees to be: the span of the identity, the traceless diagonal matrices, and the symmetric / antisymmetric off-diagonal subspaces, so that $\mcm_{{\rm AI},\mcw}$ acts as a scalar on each of them (the irreducibility and pairwise-non-isomorphic nature of these   subspaces  follows  from a ``frame potential argument''~\cite{mele2023introduction}). Interestingly,  the above direct calculation shows that the scalars corresponding to the symmetric and antisymmetric off-diagonal subspaces are in fact equal, a degeneracy which is invisible to this  argument. 

\section{Linear-depth designs over the symplectic group}\label{sec:symp}
{In this appendix we prove that $\varepsilon$-approximate relative-error $k$-designs~\cite{mele2023introduction} over the symplectic group can form in depth $\mco(nk+\log(1/\varepsilon))$, using a natural one-dimensional brickwork circuit model (see Figure~\ref{fig:symp}, and cf.   Refs.~\cite{garcia2024architectures,west2024random}). Our construction is heavily based on the approach of Ref.~\cite{baer2026random}, which proves analogous results for the unitary and orthogonal groups. As it is known that for $k\geq 2$ symplectic group designs cannot form  in sublinear (in the number $n$ of qubits) depth, this scaling is optimal with respect to $n$~\cite{west2026ambient}.  }  \\

{Let us begin by recalling some elementary properties of the Pauli and symplectic groups, and the symplectic Lie algebra. We will make use of the standard identification of $n$-qubit Pauli strings with $(2n)$-bit vectors;  for $v=(x,z)\in\mathbb F_2^{2n}$ we will write $P_v$ for the corresponding Pauli string. For a second string $w=(x',z')$ we then have the usual bilinear form $b(v,w)=x\cdot z'+z\cdot x'$ which captures the commutation structure of the Pauli group, $P_v P_w = (-1)^{b(v,w)}P_w P_v$.  It is also convenient to introduce $q(v)=x\cdot z+x_1+z_1$ (the singling out of the first qubit comes from the fact that the canonical symplectic form may be written as $J=-iY_1$).  It is then easy to see both that  $P_v^T=(-1)^{x\cdot z}P_v$ and $P_vJ=(-1)^{x_1+z_1}JP_v$. As an operator $X$ belongs to the symplectic Lie algebra $\mathfrak{sp}(d/2)$ precisely when $X^\mst J + JX=0$~\cite{fulton1991representation}, it follows that $iP_v\in\mathfrak{sp}(d/2)$ exactly when $q(v)=1$. We will denote this set of bitstrings by $\mcv=\{v:q(v)=1\}$. As  the corresponding Pauli strings form a basis of $\mfsp(d/2)$~\cite{garcia2024architectures,west2025graph}, which is   of dimension $d(d+1)/2$~\cite{fulton1991representation}, we then have $|\mathcal V|= d(d+1)/2$.  \\ }

{We are going to prove that the ensemble of unitaries that we will construct form  designs over the symplectic group using \textit{spectral methods}, which is a common approach in such situations~\cite{harrow2009random,brandao2016local,haferkamp2021improved,west2026strong,haah2024efficient,haferkamp2022random,baer2026random,chen2024incompressibility}. At a high level, and having fixed some representation in which the symplectic group is acting, this involves proving bounds on the spectral gap (i.e., one minus the largest non-unital eigenvalue) of the moment operator of our considered ensemble minus the moment operator corresponding to   the Haar measure $\mu_{\mbsp(d/2)}$. Essentially, the idea is that this uniform moment operator (being a projector onto the $\mbsp(d/2)$-invariant subspace~\cite{mele2023introduction}) has eigenvalues in $\{0,1\}$, but the moment operator of (a single ``step'') of our ensemble will generically have eigenvalues strictly between zero and one. Taking (say) $\ell$ steps essentially raises those eigenvalues to the $\ell$\textsuperscript{th} power, and we wish to take enough steps so as to drive the non-unital eigenvalues to zero (thereby ``looking like'' the uniform moment operator). Evidently, the number of steps required depends on how close the largest non-unital eigenvalue is to one, whence the common preoccupation with spectral gaps.\\  }

{It is useful to introduce a little more notation (largely following that used in  Ref.~\cite{baer2026random}). For an arbitrary finite-dimensional unitary representation of $\mbsp(d/2)$, we write 
\begin{equation}\label{eq:mrho}
M(\rho)=\expect_{v\sim \mcv,\,\vartheta\sim S^1} \rho(\exp(i\vartheta P_v))
\end{equation}
for the moment operator of the ensemble of unitaries given by randomly selecting a Pauli $iP_v\in\mfsp(d/2)$ and angle $\vartheta\in S^1$, and applying the corresponding represented group element. We also write $\Pi_{\rho}=\mbe_{g\sim \mu_{\mbsp(d/2)}} \rho(g)$ for the projector onto the trivial subspace of the representation $\rho$.  The following lemma gives a constant lower bound (in fact, $1/16$) on the spectral gap of the ensemble of unitaries whose moment operator is given by Eq.~\eqref{eq:mrho} (that is,  randomly select a Pauli in the symplectic Lie algebra and then apply the corresponding represented group element (rotated by a random $\vartheta$)):}

\begin{restatable}{lem}{symppauli}\label{lem:symppauli}
{For every finite-dimensional unitary representation  $\rho$ of $\mbsp(d/2)$   we have $\left\|M(\rho)-\Pi_{\rho}\right\|_\infty\le  1- 1/16$.
}
\end{restatable}
\begin{proof}
{We are going to follow very closely some ideas from Ref.~\cite{baer2026random}. We let $\Pi_v=\expect_{\vartheta\sim S^1} \rho(\exp( i\vartheta P_v))$,   set $ F_v=\id-\Pi_v$, and define the ``Hamiltonian'' $H=\sum_{v\in\mathcal V}F_v$. For our purposes, the relevance of $H$ is twofold. First, it is easy to see that we have $M(\rho)=\id-H/|\mcv|$. Second, the   kernel of $H$ is precisely the trivial subrepresentation   of $\rho$! Indeed, it is immediate that the invariant subspace of $\rho$ is contained within $\ker H$; in the other direction, suppose $x\in \ker H$ for some vector $x$ in the representation space of $\rho$. Because the $F_v$ are  positive semidefinite, $x$ must be in $\ker F_v$ for each $v$ individually, from which it follows immediately that $x$ is invariant under $\rho(\exp( i\vartheta P_v))$ for every $\vartheta$ and $P_v$. But the $iP_v$ span $\mfsp(d/2)$, so that by the connectedness of $\mbsp(d/2)$ any $g\in\mbsp(d/2) $ may be written as $g=\exp(\sum_v c_v iP_v) =  \lim_{r\to\infty}\big(\prod_v \exp(c_viP_v/r)\big)^r$, from which the invariance of $x$ under each individual factor   implies its invariance under the arbitrary $g$.\\ }

{Coming back to $H$ presently, let us consider the graph with vertices in $\mcv$, and an edge between $v,w\in\mathcal V$ when $b(v,w)=1$ (i.e. their Paulis anticommute). For any $v\in\mathcal V$, its   degree in this graph is then 
\begin{align*}
D&=\sum_{w\in \mcv} \frac{1-(-1)^{b(v,w)}}{2}\\
&=\sum_{w\in \mbf_2^{2n}} \frac{(1-(-1)^{q(w)})}{2}\frac{(1-(-1)^{b(v,w)})}{2}\\
&=\frac14\Big(d^2-\Big(\sum_w(-1)^{q(w)}\Big)+(-1)^{q(v)}\Big(\sum_w(-1)^{q(w)}\Big)\Big)\\
&=\frac{d(d+2)}4,
\end{align*}
where to get the last few lines a few things have happened. First, we have used the readily verified $q(w)+b(v,w)=q(v+w)+q(v)$ (with everything happening mod 2), and the fact that as $w$ ranges over $\mbf_2^{2n}$, so does $v+w$. Next, $v$ anticommutes with exactly half of the Pauli strings, so $\sum_w(-1)^{b(v,w)}=0$; finally, $\sum_w(-1)^{q(w)}=d^2-2|\mcv|=-d$. \\
}

{Now,  if $v$ and $w$ are neighbours in the graph, then $q(v+w)=q(w)+b(v,w)+q(v)=1+1+1=1$ (everything still mod 2) and $\{v,w,v+w\}$ forms a triangle $T$ in the graph.  Evidently, each edge lies in exactly one such triangle and each vertex in $D/2$ such triangles. Moreover, the three Paulis corresponding to a triangle then satisfy the usual Pauli relations and generate a copy of $\mbsu(2)$ in $\mbsp(d/2)$. Effectively, therefore, if we restrict our Hamiltonian to only include the terms from such a triangle, i.e. consider $H_T=\sum_{v\in T}F_v$, we are faced with the problem we are trying to solve, except for $\mbsu(2)$ instead of  $\mbsp(d/2)$. The $\mbsu(2)$ case may be directly understood~\cite{baer2026random,haah2024efficient}, and  e.g.   Lemma 3.13 of Ref.~\cite{baer2026random}  immediately gives the bound $ H_T^2\succeq\frac54 H_T$ (where $A\succeq B$ denotes that $A-B$ is positive semidefinite). Putting together these observations leads to  
\begin{align*}
    H^2 =\sum_{v,w} F_v F_w&= H +  \sum_{v\neq w\,:\,b(v,w)=0} F_v F_w + \sum_{v\neq w\,:\,b(v,w)=1} F_v F_w \\
    &\succeq H  + \sum_{v\neq w\,:\,b(v,w)=1} F_v F_w \\
    &= H  +  \sum_{T} (H_T^2 - H_T) \\
    &\succeq H  + \frac 14 \sum_{T}  H_T \\
    &= (1 + {D}/{8})H\\
    &= (1 + {d(d+2)}/{32})H
\end{align*}
where we have used that if $P_v$ and $P_w$ commute then $F_v F_w=(F_v F_w)^2\succeq 0$,  that every pair $(v,w)$ for anticommuting $P_v,\ P_w$ appears in exactly one triangle $T$, and that every vertex appears in $D/2$ triangles. It follows that the eigenvalues of $H$ are either zero or at least $(1 + {d(d+2)}/{32})$. Let us finally turn to the target of the lemma, the eigenspectrum of
\begin{equation*}
    M(\rho)-\Pi_\rho = \id - \frac{H}{|\mcv|}- {\rm proj}_{\ker H}= \id - \frac{2H}{d(d+1)} - {\rm proj}_{\ker H}\,.
\end{equation*}
Evidently this operator  vanishes on $\ker H$; on $(\ker H)^\perp$ its eigenvalues are by the above calculation   at most $ 1 - 2/(d(d+1))(1 + {d(d+2)}/{32})< 1 - 1/16$. They are also manifestly nonnegative, which establishes the lemma.
}
\end{proof}

\begin{figure}
    \centering
    \begingroup
\definecolor{boundaryblue}{RGB}{77,128,173}
\definecolor{bulkumber}{RGB}{104,81,154}
\definecolor{wiregray}{RGB}{88,94,102}
\begin{tikzpicture}[
  x=1cm,y=0.7cm,
  font=\small,
  wire/.style={draw=wiregray,line width=0.55pt},
  gate/.style={
    draw,line width=0.8pt,rounded corners=1.5pt,
    minimum width=1.28cm,minimum height=1.05cm,
    inner sep=2pt,font=\small
  },
  boundary/.style={gate,draw=boundaryblue,fill=boundaryblue!10},
  bulk/.style={gate,draw=bulkumber,fill=bulkumber!11},
  layerlabel/.style={font=\footnotesize,text=wiregray},
  roundbrace/.style={
    draw=wiregray,line width=0.6pt,
    decorate,decoration={brace,amplitude=4pt}
  }
] 
  \foreach \j in {1,...,6} {
    \pgfmathsetmacro{\wirey}{1-\j}
    \node[anchor=east] at (0.0,\wirey) {};
    \draw[wire] (0.18,\wirey) -- (7.95,\wirey);
    \node[text=wiregray] at (8.55,\wirey) {$\cdots$};
    \draw[wire] (9.15,\wirey) -- (13.05,\wirey);
  }

  \foreach \oddx/\evenx/\roundlabel in {1.3/3.05/1,5.15/6.9/2,10.2/11.95/\ell} {
    \node[boundary] at (\oddx,-0.5) {$\operatorname{Sp}(2)$};
    \foreach \centery in {-2.5,-4.5} {
      \node[bulk] at (\oddx,\centery) {$\operatorname{Cl}_{\mathbb R}(2)$};
    }
    \foreach \centery in {-1.5,-3.5} {
      \node[bulk] at (\evenx,\centery) {$\operatorname{Cl}_{\mathbb R}(2)$};
    }
    \node[layerlabel] at (\oddx,0.68) {odd layer};
    \node[layerlabel] at (\evenx,0.68) {even layer};
    \draw[roundbrace]
      ({\oddx-0.64},1.16) -- ({\evenx+0.64},1.16)
      node[midway,above=7pt] {round $\roundlabel$};
  }

\end{tikzpicture}
\endgroup
    \caption{Each ``round'' of our brickwork circuit architecture consists of a Haar random symplectic matrix on the first two qubits, and then two layers of random two qubit real-valued Clifford operators, alternating as depicted in this six-qubit example. }
    \label{fig:symp}
\end{figure}

\noindent 
{We can now prove:}

\begin{restatable}{thm}{thmsymp} 
The brickwork architecture of Figure~\ref{fig:symp} converges to an $\varepsilon$-approximate relative error $k$-design  over the $n$-qubit symplectic group in $\ell=\mco(nk\log 4 + \log(1/\varepsilon))$ layers.
\end{restatable}

\begin{proof}

{It is helpful to introduce some objects that probe a little more explicitly the two aspects of our  ensemble (and in particular expose its asymmetry between the first and remaining $n-1$ qubits). To that end, having fixed a representation $\rho$ we let $A$ denote the projector on the invariant subspace of the restriction of $\rho$ to the subgroup given by the orthogonal (i.e. real-valued) Clifford group acting on the last $n-1$ qubits (recall Figure~\ref{fig:symp}), and $B$ denote   the projector on the invariant subspace of the restriction of $\rho$ to the subgroup given by the symplectic group   acting on the first two qubits. It is then useful to note that, for any symplectic Pauli $Q$, we have $
ABA\preceq A\Pi_QA$. This follows from the characterisation of the Pauli  strings in $\mfsp(d/2)$ as $\{\{X, Y, Z\} \ot P_s\} \cup  \{I \ot P_a\}$, where $P_s$ denotes a Pauli string satisfying $P_s^\mst=P_s$, and similarly $P_a^\mst=-P_a$~\cite{garcia2024architectures}; indeed, any such $Q$ may therefore be conjugated by a real Clifford $C$ acting  on the  last $n-1$ qubits into one of $
I_1Y_2,\ a_1Z_2,\ a_1, $ where $ a\in\{X,Y,Z\}.$ It follows that $B\preceq \Pi_{CQC\ad} =\rho(C)\Pi_Q \rho(C\ad)$ (as being invariant under the induced copy of $S^1$ inside $\mbsp(2)$ is a weaker   requirement than being invariant under all of $\mbsp(2)$); averaging over all real Cliffords acting  on the  last $n-1$ qubits then gives $ABA\preceq A\Pi_Q A$. Averaging yet further, this time over $Q$, and using Lemma~\ref{lem:symppauli}  gives in turn $ABA \preceq AM(\rho)A \preceq \Pi_\rho+\frac{15}{16}(A-\Pi_\rho)$, where the last inequality follows from the fact that $M(\rho)$ respects the decomposition $\rho\cong {\rm im}\, \Pi_\rho \oplus \ker \Pi_\rho $, with respect to which it takes the form 
\begin{equation*}
M(\rho)=\begin{pmatrix}\id&0\\0&M_\perp
\end{pmatrix},
\end{equation*}
where $M_\perp$ is constrained by Lemma~\ref{lem:symppauli} to be less than $(15/16) \id$, whence $M(\rho)\preceq \Pi_\rho +\frac{15}{16}(\id-\Pi_\rho )$. An immediate consequence of all this that we will use later is that
\begin{equation}\label{eq:aba}
    ABA -\Pi_\rho \preceq \frac{15}{16}(A-\Pi_\rho)\preceq \frac{15}{16} \id
\end{equation}
\\
}

{Next, let us consider replacing the action of the global random real Clifford on the final $n-1$ qubits with the   brickwork architecture of Figure~\ref{fig:symp}. We will call the moment operator of a single ``round'' i.e. an     even and an odd layer of two-qubit random real Cliffords $T$. As $T$ corresponds to averaging over a strict subset of the  group averaged over by $A$, we immediately have $TA=AT=A$. Less obviously, Theorem 2.49 of Ref.~\cite{baer2026random} further gives $\|T-A\|\le 143/144 =:q$. A full round of our brickwork architecture therefore   has   moment operator $TB$,  and our   goal is to characterise the spectral gap of $TB-\Pi_\rho$. It is helpful to again consider the decomposition $\rho\cong {\rm im}\, \Pi_\rho \oplus \ker \Pi_\rho $, with respect to which $T,\,B$, and $\Pi_\rho$ are all block diagonal. First, on $ {\rm im}\, \Pi_\rho$ we have $T=B=\Pi_\rho=\id_{{\rm im}\, \Pi_\rho}$, so that $(TB-\Pi_\rho)\rvert_{{\rm im}\, \Pi_\rho}=0$. So, consider an arbitrary  $\psi\in\ker \Pi_\rho$. Using the readily verified $T=A+(T-A)(I-A) $, and the fact that $A$ being a projector implies (for any vector $x$) that $\|x\|^2 = \|Ax\|^2+\|(\id-A)x\|^2$, we calculate
\begin{align}
\|(TB-\Pi_\rho)\psi\|^2=\|TB\psi\|^2&=\|(A+(T-A)(\id-A))B\psi\|^2\nonumber\\
&\le \|AB\psi\|^2+q^2\|(\id-A)B\psi\|^2\nonumber\\
&=q^2\|B\psi\|^2+(1-q^2)\|AB\psi\|^2\nonumber\\
&\le q^2\|\psi\|^2+(1-q^2)\|AB\psi\|^2\nonumber\\
&\le\left(1-\frac{1-q^2}{16}\right)\|\psi\|^2.\label{eq:tbl}\\
&=\left(\frac{331489}{331776}\right)\|\psi\|^2.\label{eq:tbll}
\end{align}
To see the last step, note that $AB\psi = (AB-\Pi_\rho)\psi$, and that $(AB-\Pi_\rho)(AB-\Pi_\rho)^\dagger=ABA-\Pi_\rho$, so that from the generic operator inequality $\|X\|^2=\|XX^\dagger\|$ we have $\|AB-\Pi_\rho\|^2=\|ABA-\Pi_\rho\|$. From Eq.~\eqref{eq:aba} it then follows that $\|AB\psi\|^2 \leq (15/16)\|\psi\|^2  $, from which we obtain Eq.~\eqref{eq:tbl}. We will write $\Delta=1-\sqrt{331489/331776}$ for the lower bound on the spectral gap of a single round of our brickwork construction implied by Eq.~\eqref{eq:tbll}. \\
}

{To finish the story, we note that taking $\ell$ rounds of our brickwork construction leads to an ensemble $\mce_\ell$ with $\|M_{\mce_\ell}(\rho) -\Pi_\rho\|\leq (1-\Delta)^\ell\leq \exp(-\ell\Delta)$. In general, there is an overhead by a factor of up to $d^{2k}=4^{nk}$ when moving from an approximate $k$-design with respect to the spectral norm to one of our   target, the relative error norm~\cite{li2024designs}, from which we see that we will have an $\varepsilon$-approximate relative error design  over the $n$-qubit symplectic group if $4^{nk}\exp(-\ell\Delta)\leq \varepsilon$, which gives $\ell\Delta\geq nk\log 4 + \log(1/\varepsilon)$. As $\Delta$ is a constant, we are done. }
\end{proof}

\end{document}